\documentclass[conference,letterpaper]{IEEEtran}

\usepackage[utf8]{inputenc} 
\usepackage[T1]{fontenc}
\usepackage{url}
\usepackage{ifthen}
\usepackage{cite}
\usepackage[cmex10]{amsmath} % Use the [cmex10] option to ensure complicance
\usepackage{amsfonts} 
\usepackage{tabularx}
\usepackage{longtable}
\usepackage{arydshln} 
\usepackage{amssymb} 
\usepackage{amsthm} 
\usepackage{mathtools} 
\usepackage{xcolor} 

\DeclarePairedDelimiter{\ceil}{\lceil}{\rceil}

\newtheorem{Theorem}{Theorem}
\newtheorem{Lemma}{Lemma}
\newtheorem{Remark}{Remark}
\newtheorem{Example}{Example}

\newtheorem{Proposition}{Proposition}

\begin{document}
\title{On Decentralized Linearly Separable Computation With the Minimum Computation Cost}  

% %%% Single author, or several authors with same affiliation:
% \author{%
%  \IEEEauthorblockN{Andrew R.~Barron}
%  \IEEEauthorblockA{Department of Statistics and Data Science\\
%                    Yale University\\
%                    New Haven, CT, USA\\
%                    Email: andrew.barron@yale.edu}
% }

%%% Several authors with up to three affiliations:
\author{
	\IEEEauthorblockN{%
		Haoning Chen\IEEEauthorrefmark{1},
 Minquan Cheng\IEEEauthorrefmark{2}, Zhenhao Huang\IEEEauthorrefmark{1},
and 		Youlong Wu\IEEEauthorrefmark{1}
	}
	%\IEEEauthorblockA{}
	  \IEEEauthorblockA{\IEEEauthorrefmark{1} ShanghaiTech University, Shanghai, China, \\  \{chenhn, huangzhh, wuyl1\}@shanghaitech.edu.cn }
	  \IEEEauthorblockA{\IEEEauthorrefmark{2} Guangxi Normal University, Guilin, China, mqcheng@gxnu.edu.cn}
}

%%% Many authors with many affiliations:
% \author{%
%   \IEEEauthorblockN{Andrew R.~Barron\IEEEauthorrefmark{1},
%                     Claude E.~Shannon\IEEEauthorrefmark{2},
%                     David Slepian\IEEEauthorrefmark{2},
%                     and Jacob Ziv\IEEEauthorrefmark{2}\IEEEauthorrefmark{3}}
%   \IEEEauthorblockA{\IEEEauthorrefmark{1}%
%                    Department of Statistics and Data Science, Yale University, New Haven, CT, USA,
%                     andrew.barron@yale.edu}
%   \IEEEauthorblockA{\IEEEauthorrefmark{2}%
%                     Bell Telephone Laboratories, Inc.,
%                     Murray Hill, NJ, USA,
%                     \{csh,dsl,jz\}@bell-labs.com}
%   \IEEEauthorblockA{\IEEEauthorrefmark{3}%
%                     Department of Electrical Engineering, Technion---Institute of Technology, Haifa, Israel,
%                     jz@ee.technion.ac.il}
% }

\maketitle

%%%%%%
%% Abstract: 
%% If your paper is eligible for the student paper award, please add
%% the comment "THIS PAPER IS ELIGIBLE FOR THE STUDENT PAPER
%% AWARD." as a first line in the abstract. 
%% For the final version of the accepted paper, please do not forget
%% to remove this comment!
%%

\begin{abstract}
The distributed linearly separable computation problem finds extensive applications across domains such as distributed gradient coding, distributed linear transform, real-time rendering, etc. In this paper, we investigate this problem in a fully decentralized scenario, where $\mathsf{N}$ workers collaboratively perform the computation task without a central master. Each worker aims to compute a linearly separable computation that can be manifested as $\mathsf{K}_{\mathrm{c}}$ linear combinations of $\mathsf{K}$ messages, where each message is a function of a distinct dataset. We require that each worker successfully fulfill the task based on the transmissions from any $\mathsf{N}_{\mathrm{r}}$ workers, such that the system can tolerate any $\mathsf{N}-\mathsf{N}_{\mathrm{r}}$ stragglers. We focus on the scenario where the computation cost (the number of uncoded datasets assigned to each worker) is minimum, and aim to minimize the communication cost (the number of symbols the fastest $\mathsf{N}_{\mathrm{r}}$ workers transmit). We propose a novel distributed computing scheme that is optimal under the widely used cyclic data assignment. Interestingly, we demonstrate that the side information at each worker is ineffective in reducing the communication cost when $\mathsf{K}_{\mathrm{c}}\leq  {\mathsf{K}}\mathsf{N}_{\mathrm{r}}/{\mathsf{N}}$, while it helps reduce the communication cost as $\mathsf{K}_{\mathrm{c}}$ increases.    
\end{abstract} 


\section{Introduction} 
\label{Sec:Introduction} 
Recently distributed computing has garnered substantial attention \cite{AWS, Azure, Cloud}, due to its capacity to concurrently process intricate computational tasks across numerous nodes, thereby accelerating the overall computation speed. Nevertheless, the efficacy of distributed computing is adversely affected by challenges stemming from both limited communication bandwidth and the presence of straggling workers \cite{surveyCDC}. Coding techniques were initially utilized to address the aforementioned two problems, specifically in reducing the communication cost \cite{CDC} and mitigating the impact of stragglers \cite{Lee'18}.

% The distributed machine learning model has been widely studied in the literature. However, as stated in \cite{Kairouz'21}, in more collaborative learning scenarios, an omniscient central server may not be feasible or desirable \cite{Van'17}. Moreover, when the number of clients is very large, there could be a communication traffic jam on the central server, since all users have to communicate with it in each iteration \cite{Lian'17}. Motivated by this, various works considered distributed machine learning problems in a fully decentralized scenario \cite{Bellet'18, Yakimenka'22, Xu'22}, where the client communicates with each other in a peer-to-peer fashion, without the coordination of a central server.  

Distributed linearly separable computation is a specific distributed computing framework widely studied over the canonical centralized, single-master coded computing system \cite{m=1, m>1TCOM, m>1, secure}, where 
a master node aims to compute a function of $\mathsf{K}$ datasets, represented as $\mathsf{K}_{\mathrm{c}}$ linear combinations of $\mathsf{K}$ messages, where each message corresponds to an individual function of a distinct dataset. Such a computation task structure encompasses various practical applications, including but not limited to distributed gradient descent \cite{GC, Ye'18, Raviv'18, Halbawi'18, Xu'21, AdaptiveGC}, distributed linear transform \cite{Short-Dot}, real-time rendering \cite{Moller'19}, etc. 
%the master wishes to compute a linearly separable function $f$ of $\mathsf{K}$ datasets $D_1, \ldots, D_{\mathsf{K}}$, i.e.,  
%\begin{align} \label{task} 
%&f(D_1, \ldots, D_{\mathsf{K}})=g(f_1(D_1), \ldots, f_{\mathsf{K}}(D_{\mathsf{K}}))=g(W_1, \ldots, %W_{\mathsf{K}}),    
%\end{align} 
%where the message $W_k=f_k\left(D_k\right)\in\mathbb{F}_{\mathsf{q}}^{\mathsf{L}}$, $k\in\left[\mathsf{K}\right]$ is computed by the partial function $f_k\left(\cdot\right)$ taking $D_k$ as input. $g\left(W_1, \ldots, W_{\mathsf{K}}\right)=\mathbf{F}\left[W_1; \ldots; W_{\mathsf{K}}\right]$ represents $\mathsf{K}_{\mathrm{c}}$ linear combinations of $W_1, \ldots, W_{\mathsf{K}}$, where the elements in $\mathbf{F}\in\mathbb{F}_{\mathsf{q}}^{\mathsf{K}_{\mathrm{c}}\times\mathsf{K}}$ are uniformly i.i.d. over $\mathbb{F}_{\mathsf{q}}$. Each of the $\mathsf{N}$ workers contributes to the computation of \eqref{task} by transmitting coded messages to the master, derived from their respective assigned datasets. 

Nevertheless, as articulated in \cite{Maddah-Ali'23}, the master would be a bottleneck for scalability in a distributed computing system.  A substantial increase in the number of workers may lead to communication congestion at the master, given that each worker is required to communicate with it in every iteration \cite{Lian'17}. To overcome this, decentralized distributed computations where computing nodes exchange information in a decentralized fashion without the help of a master node have been widely studied \cite{CDC, Jeong'20}. Unfortunately, 
the existing coded computing schemes are not designed delicately for linearly separable computation. They are either infeasible in the linearly separable computation problem, or unable to fully exploit the linear algebra property to minimize the communication costs.  

%It is worth noting that although the centralized schemes in \cite{m=1, m>1TCOM, m>1, secure} could be extended to the decentralized scenarios (by letting all distributed workers perform the same role as the master node), they ignore workers' side information computed from their local datasets and thus may incur redundant computation load. 

%%(by letting all distributed workers perform the same role as the master node) 

It is worth noting the centralized schemes in \cite{m=1, m>1TCOM, m>1} could be easily extended to the decentralized scenarios. This can be achieved by simply letting  
 each worker transmit the coded messages, which were originally intended for the master in \cite{m=1, m>1TCOM, m>1}, to the other workers, and computing the $\mathsf{K}_{\mathrm{c}}$ linear combinations at the workers. However, this approach is highly sub-optimal as the master in \cite{m=1, m>1TCOM, m>1} does compute any message locally. In contrast, in our considered problem, each worker can compute some messages locally before receiving the transmissions from other workers. Consequently, directly applying the scheme in \cite{m=1, m>1TCOM, m>1} fails to leverage the side information from each worker, resulting in unnecessary communication costs. %In light of this, in this paper, we introduce a novel approach to better exploit the locally computed messages of each worker and minimize the communication cost. 

%Moreover,       An benchmark scheme for the considered problem is to let each worker transmit the coded messages intended for the master in \cite{m=1} to the other workers. However, this approach is highly sub-optimal as the master in \cite{m=1} does compute any message locally. In contrast, in our considered problem, each worker can compute some messages locally before receiving the transmissions from other workers. Consequently, directly applying the scheme in \cite{m=1} fails to leverage the side information from each worker, resulting in unnecessary communication costs. In light of this, in this paper, we introduce a novel approach to better exploit the locally computed messages of each worker and minimize the communication cost. 

%the coded computing scheme in \cite{CDC} is designed for the general MapReduce tasks, not limited to linearly separable computation, and although the schemes in \cite{m=1} could be applied to decentralised scenarios with 

% of 

%Moreover, in many pratical  master's role as the sole entity responsible for decoding all results could potentially become a bottleneck, impacting the robustness of the system. Motivated by this, the works in \cite{CDC, Jeong'20} studied a masterless coded computing system in which the workers are entrusted with the execution of computation tasks.  %Additionally, the work in \cite{multiuser} investigates the distributed linearly separable computation problem within a multi-master framework. 

Motivated by the facts above, in this paper we study the decentralized linearly separable computation problem in a fully decentralized scenario, where $\mathsf{N}$ workers connect with each other through a shared and noiseless multicast link,  with the presence of $\mathsf{N}-\mathsf{N}_{\mathrm{r}}$ stragglers. Each worker wishes to compute a  linearly separable computation that can be expressed as $\mathsf{K}_{\mathrm{c}}$ linear combinations of $\mathsf{K}$ messages, with each message being generated from a distinct dataset.  To perform such  a linearly separable computation task, the $\mathsf{N}$ workers first are assigned  $\mathsf{K}$ datasets, and then compute messages from the assigned dataset and exchange information with each other through the shared link. Finally, each worker recovers the $\mathsf{K}_{\mathrm{c}}$ desired linear combinations from any $\mathsf{N}_{\mathrm{r}}$ responding workers, such that the system can tolerate $\mathsf{N}-\mathsf{N}_{\mathrm{r}}$ stragglers.   Our goal is to find the optimal communication cost $\mathsf{R}$ (the normalized number of symbols transmitted by the responding workers) under the cyclic assignment\footnote{The cyclic assignment is a simple data assignment strategy widely used in related works studying the distributed linearly separable problem. It is unlimited by system parameters and independent of the specific task function. }, given the minimum computation cost (the number of uncoded datasets assigned to each worker).  The main contributions of this paper are summarized as follows. 

\begin{itemize} 
\item We propose a novel distributed computing scheme for the considered decentralized linearly separable computation problem, by leveraging the side information at each worker to minimize the communication cost. In particular, based on the intersection of the linear spaces of the computed messages and the demanded linear combinations, each worker sends the \emph{minimum} number of linear   combinations of messages uniformly at random such that all the workers can decode   the $\mathsf{K}_{\mathrm{c}}$ desired linear combinations based on the local   and delivered messages. 

%With a highly non-trivial proof, we formally show that with high probability the demanded linear combinations of each worker lie in the linear span of its known messages and the transmissions from other workers.   
%When the number of linear combinations is small, we show that it suffices to let each worker construct its transmitting signal in the same manner as that of \cite{m=1}. When the number of linear combinations is large, we introduce a novel approach for each worker to construct it transmitting signal. In particular, from the intersection of the linear spaces of the computed messages and the demanded linear combinations, each worker selects some linear combinations of messages uniformly at random and sends them to the other workers. 
%The proof of correctness of our proposed scheme is highly non-trivial, the main challenge is to show that despite the random selections of transmitting signals, with high probability each worker is still capable of decoding the desired computation task, although the messages that each worker knows are different, and any set of \textnormal{$\mathsf{N}_{\mathrm{r}}$} workers could be the responding workers. 

\item Compared with the benchmark scheme \cite{m=1}, our proposed scheme achieves a strictly smaller communication cost when the number of linear combinations is relatively large, i.e., $\mathsf{K}_{\mathrm{c}}>  {\mathsf{K}}\mathsf{N}_{\mathrm{r}}/{\mathsf{N}}$. This improvement mainly comes from our scheme's adept utilization of locally computed messages from each worker, thereby minimizing communication overhead. %the scheme in \cite{ \cite{m=1}} while

\item We analyze the converse bound for our considered problem and establish the optimality of our proposed scheme under the cyclic data assignment. In other words, when $\mathsf{K}_{\mathrm{c}}>  {\mathsf{K}}\mathsf{N}_{\mathrm{r}}/{\mathsf{N}}$ and under the cyclic assignment,  our scheme is optimal while the benchmark scheme \cite{m=1} is not. Surprisingly, we prove that when $\mathsf{K}_{\mathrm{c}}\leq  {\mathsf{K}}\mathsf{N}_{\mathrm{r}}/{\mathsf{N}}$, the benchmark scheme \cite{m=1} is still optimal, indicating the side information at workers is useless in reducing the communication cost when $\mathsf{K}_{\mathrm{c}}$ is relatively small compared to $\mathsf{N}_{\mathrm{r}}$.  %We think this is mainly because that  $\mathsf{K}_{\mathrm{c}}$ provides 

%This is mainly because that workers are unaware in advance of which among them will become the responding workers. In case of the worst-case scenario, each worker has to transmit as many symbols as in \cite{m=1} to the other workers.%indicating that   the side information at workers is useless in decoding the  the communication cost

\end{itemize}

\;\;\emph{Notations:} %We let bold capital letters denote matrices, bold lower-case letters denote vectors,  calligraphic letters denote sets, and sans-serif letters denote system parameters. 
Define $\left[x:y\right]:=\left\{x, x+1, \ldots, y\right\}$, $\left[x:y\right):=\left\{x, x+1, \ldots, y-1\right\}$, $\left(x:y\right]:=\left\{x+1, x+2, \ldots, y\right\}$, $\left(x:y\right):=\left\{x+1, x+2, \ldots, y-1\right\}$, $\left[n\right]:=\left\{1, 2, \ldots, n\right\}$. Denote $\lvert\cdot\rvert$ as the cardinality of a set, and let $\mathcal{A}\setminus\mathcal{B}:=\left\{x\in\mathcal{A}\mid x\not\in\mathcal{B}\right\}$. $\mathbb{F}_{\mathsf{q}}$ represents a finite field with order $\mathsf{q}$. For a matrix $\mathbf{A}$, $\mathbf{A}^T$ and $\mathbf{A}^{-1}$ represent its transpose and inverse, respectively; $\mathcal{N}\left(\mathbf{A}\right)$ represents the null space of $\mathbf{A}$, and $\mathcal{C}\left(\mathbf{A}\right)$ represents the column space of $\mathbf{A}$. $\mathrm{span}\left\{\mathbf{a}_1, \ldots, \mathbf{a}_n\right\}$ represents the linear span of the vectors $\mathbf{a}_1, \ldots, \mathbf{a}_n$, and $\dim\mathcal{V}$ represents the dimension of the subspace $\mathcal{V}$. $\text{Mod}\left(x, y\right)\in\left\{1, \ldots, y\right\}$ represents the remainder after dividing $x$ by $y$, where we let $\text{Mod}\left(x, y\right)=y$ if $y$ divides $x$. 
% The elements of a set $\mathcal{S}=\left\{\mathcal{S}_1, \ldots, \mathcal{S}_{\lvert\mathcal{S}\rvert}\right\}$ are sorted in an increasing order.  

\section{System Model and Problem Formulation} 
\label{Sec:Model}
We formulate a $(\mathsf{K}, \mathsf{N}, \mathsf{N}_{\mathrm{r}}, \mathsf{K}_{\mathrm{c}}, \mathsf{M})$ distributed linearly separable computation problem over a fully decentralized network, where $\mathsf{N}$ workers connect with each other  through a shared, noiseless device-to-device (D2D)  link with the presence of $\mathsf{N}-\mathsf{N}_{\mathrm{r}}$ stragglers, as depicted in Fig. 1. Each worker wishes to compute a function of $\mathsf{K}$ statistically independent datasets $D_1, \ldots, D_{\mathsf{K}}$, which is assumed to be linearly  separable from the datasets and can be written as $\mathsf{K}_{\mathrm{c}}\leq \mathsf{K}$ linear combinations of $\mathsf{K}$ messages, i.e., 
\begin{IEEEeqnarray}{rCl}
&&f\left(D_1, \ldots, D_{\mathsf{K}}\right)=g\left(f_1\left(D_1\right), \ldots, f_{\mathsf{K}}\left(D_{\mathsf{K}}\right)\right) \nonumber 
\\&&=g\left(W_1, \ldots, W_{\mathsf{K}}\right)=\mathbf{F}\left[W_1; \ldots; W_{\mathsf{K}}\right]=\left[F_1; \ldots; F_{\mathsf{K}_\text{c}}\right],   
\end{IEEEeqnarray}
where the $k$-th message $W_k=f_k\left(D_k\right)$, $k\in\left[\mathsf{K}\right]$ is generated from the (generally non-linear and computationally hard) sub-function $f_k\left(\cdot\right)$ taking $D_k$ as input, and $F_i$, $i\in\left[\mathsf{K}_{\mathrm{c}}\right]$ is the $i$-th row of $\mathbf{F}\left[W_1; \ldots; W_{\mathsf{K}}\right]$. As in \cite{m=1, m>1TCOM, m>1, secure}, Each of the $\mathsf{K}$ messages is assumed to be uniformly i.i.d. over $\mathbb{F}_{\mathsf{q}}^{1\times\mathsf{L}}$, for some sufficiently large $\mathsf{q}$, where $\mathsf{L}$ is assumed to be sufficiently large such that any sub-message division is possible. 
$\mathbf{F}\in\mathbb{F}_{\mathsf{q}}^{\mathsf{K}_{\mathrm{c}}\times\mathsf{K}}$ represents the demand matrix, with its elements uniformly i.i.d. over $\mathbb{F}_{\mathsf{q}}$. In this paper, we assume that $\frac{\mathsf{K}}{\mathsf{N}}$ is an integer.\footnote{If $\frac{\mathsf{K}}{\mathsf{N}}$ is not an integer, we could inject $\ceil[\Big]{\frac{\mathsf{K}}{\mathsf{N}}}\mathsf{N}-\mathsf{K}$ empty datasets into the system as in \cite{m=1}. } 
% Thus the computation task $f\left(D_1, \ldots, D_{\mathsf{K}}\right)$ can be seen as $\mathsf{K}_{\mathrm{c}}$ linearly independent combinations of $W_1, \ldots, W_{\mathsf{K}}$. 

\begin{figure} 
    \centering
    \label{fig:sysmod}
    \includegraphics[width=3.3in,scale=0.35]{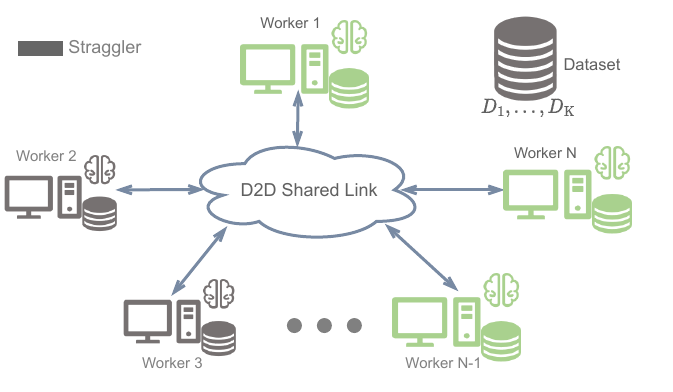}   
    \caption{{{The considered decentralized computing system.}}}
\end{figure}

The distributed computing framework is divided into the following three phases. 

\subsubsection{Data Assignment Phase} In this phase, the datasets are assigned to the workers without priorly knowing the demand matrix and stragglers' identities. 
We denote $\mathcal{Z}_n\in\left[\mathsf{K}\right]$ as the index set of datasets assigned to worker $n\in\left[\mathsf{N}\right]$, satisfying $\lvert\mathcal{Z}_n\rvert=\mathsf{M}$. Since the computation overhead of separable functions $f_1, \ldots, f_{\textsf{K}}$ is generally much higher than than the linear combinations of messages, we follow the same definition as in  \cite{m=1, m>1TCOM, m>1, secure} to 
denote $\textsf{M}$ as the \emph{computation cost}.

In this paper, we focus on a \emph{cyclic} data assignment and assume each worker knows all index sets   $\{\mathcal{Z}_n:n\in\left[\mathsf{N}\right]\}$. The cyclic assignment is easy to implement and 
%and assume assignment strategy is known 
%\;\;\emph{Cyclic Assignment:}  The cyclic assignment 
 has been widely adopted in \cite{GC, Ye'18, Raviv'18, Xu'21, AdaptiveGC, m=1, m>1TCOM, m>1}. Under the cyclic assignment, dataset $D_k$, $k\in\left[\mathsf{K}\right]$,  is assigned to workers $\text{Mod}\left(k, \mathsf{N}\right), \text{Mod}\left(k-1, \mathsf{N}\right),  \ldots$, $\text{Mod}\left(k-\mathsf{N}+\mathsf{N}_{\mathrm{r}}, \mathsf{N}\right)$. Thus the set of datasets assigned to worker $n\in\left[\mathsf{N}\right]$ is  
\begin{IEEEeqnarray}{rCl} \label{CA} 
&&\mathcal{Z}_n=\mathop{\cup}\limits_{p\in\left[0:\frac{\mathsf{K}}{\mathsf{N}}-1\right]} \big\{\text{Mod}\left(n, \mathsf{N}\right)+p\mathsf{N}, \text{Mod}\left(n+1, \mathsf{N}\right)+p\mathsf{N}, \ldots,  \nonumber 
\\&&\;\;\;\;\;\;\;\;\;\;\;\;\;\;\;\;\;\;\;\;\;\;\;\;\; \text{Mod}\left(n+\mathsf{N}-\mathsf{N}_{\mathrm{r}}, \mathsf{N}\right)+p\mathsf{N}\big\}. 
\end{IEEEeqnarray} 
In this case, we have $\mathsf{M}=\frac{\mathsf{K}}{\mathsf{N}}\left(\mathsf{N}-\mathsf{N}_{\mathrm{r}}+1\right)$, which is the minimum computation cost since each dataset should be assigned to at least $\mathsf{N}-\mathsf{N}_{\mathrm{r}}+1$ workers \cite{m=1}.

%where $\mathsf{N}_{\mathrm{r}}\in\left[\mathsf{N}\right]$ and $\mathsf{m}\in\left[\mathsf{N}_{\mathrm{r}}\right]$.     
% where $\mathsf{M}$ is defined as the computation cost.            

\subsubsection{Computing Phase} 
In this phase, we assume the demand matrix $\mathbf{F}$ is known by all workers. This can be realized by broadcasting $\mathbf{F}$ to all workers and the resulting communication cost is almost negligible when $\mathsf{L}$ is sufficiently large. Each worker $n\in\left[\mathsf{N}\right]$ first computes $W_k=f_k\left(D_k\right)$ for any $k\in\mathcal{Z}_n$, then it creates a signal $X_n=\psi_n\left(\left\{W_k: k\in\mathcal{Z}_n\right\}, \mathbf{F}\right)\in\mathbb{F}_{\mathsf{q}}^{1\times\mathsf{T}_n}$, where the encoding function $\psi_n$ is given by $\psi_n: \left[\mathbb{F}_{\mathsf{q}}\right]^{\lvert\mathcal{Z}_n\rvert\mathsf{L}}\times \left[\mathbb{F}_{\mathsf{q}}\right]^{\mathsf{K}_{\mathrm{c}}\mathsf{K}}\rightarrow \left[\mathbb{F}_{\mathsf{q}}\right]^{\mathsf{T}_n}$.  
%as a function of its locally generated messages $\left\{W_k: k\in\mathcal{Z}_n\right\}$, 
Finally, worker $n\in\mathcal{A}$ sends $X_n$ to all workers in $\left[\mathsf{N}\right]\setminus\left\{n\right\}$, where $\mathcal{A}$ represents the set of responding workers with $\lvert\mathcal{A}\rvert=\mathsf{N}_{\mathrm{r}}$.

%In this paper, in accordance with the conventions established in \cite{m=1, m>1TCOM, m>1, secure}, we denote $\mathsf{M}$ as the computation cost. This designation is particularly pertinent as we assume that the computation of separable functions $f_1, \ldots, f_{\mathsf{K}}$ incurs significantly higher costs compared with the computation of the desired linear combinations of messages. 

\subsubsection{Decoding Phase} 
%Since the workers have no idea which of them will become stragglers, 
We stipulate that each worker successfully recover $g\left(W_1, \ldots, W_{\mathsf{K}}\right)$ based on the transmissions from any subset $\mathcal{A}$. In other words, the system should be able to tolerate any $\mathsf{N}-\mathsf{N}_{\mathrm{r}}$ stragglers. Worker $n\in\left[\mathsf{N}\right]$ uses $X_{\mathcal{A}}:=\left\{X_i: i\in\mathcal{A}\setminus\left\{n\right\}\right\}$ and its local messages $\left\{W_k: k\in\mathcal{Z}_n\right\}$ to recover the $\mathsf{K}_{\mathrm{c}}$ target linear combinations.  
In particular, there exists a decoding function $\phi_{\mathcal{A}, n}: \left[\mathbb{F}_{\mathsf{q}}\right]^{\sum_{i\in\mathcal{A}\setminus\left\{n\right\}} \mathsf{T}_i}\times \left[\mathbb{F}_{\mathsf{q}}\right]^{\lvert\mathcal{Z}_n\rvert\mathsf{L}}\times \left[\mathbb{F}_{\mathsf{q}}\right]^{\mathsf{K}_{\mathrm{c}}\mathsf{K}}\rightarrow \left[\mathbb{F}_{\mathsf{q}}\right]^{\mathsf{K}_{\mathrm{c}}\mathsf{L}}$ such that $\hat{g}_{\mathcal{A}, n}=\phi_{\mathcal{A}, n}\left(X_{\mathcal{A}}, \left\{W_k: k\in\mathcal{Z}_n\right\}, \mathbf{F}\right)$, for all $n\in\left[\mathsf{N}\right]$.   
          
We define the worst-case probability of error for worker $n\in\left[\mathsf{N}\right]$ as 
\begin{IEEEeqnarray}{rCl}
\varepsilon_n :=\mathop{\max}\limits_{\mathcal{A}\subseteq\left[\mathsf{N}\right]: \lvert\mathcal{A}\rvert=\mathsf{N}_{\mathrm{r}}} \Pr\left\{\hat{g}_{\mathcal{A}, n}\ne g\left(W_1, \ldots, W_{\mathsf{K}}\right)\right\}.  
\end{IEEEeqnarray} 
A computing scheme is achievable if the worst-case probability of error $\varepsilon_n\rightarrow0$ when $\mathsf{q}\rightarrow\infty$, for all $n\in\left[\mathsf{N}\right]$.  Moreover, the \emph{communication cost} is defined as 
\begin{IEEEeqnarray}{rCl} 
\mathsf{R}:=\mathop{\max}\limits_{\mathcal{A}\subseteq\left[\mathsf{N}\right]: \lvert\mathcal{A}\rvert=\mathsf{N}_{\mathrm{r}}} \frac{\sum_{n\in\mathcal{A}} \mathsf{T}_n}{\mathsf{L}}.   
\end{IEEEeqnarray} 

We denote the optimal communication cost under the cyclic assignment in \eqref{CA} as $\mathsf{R}_{\mathrm{cyc}}^*$.

\;\;\emph{Benchmark Scheme \cite{m=1}:} We can let each worker $n\in\left[\mathsf{N}\right]$ construct $X_n$ in the same manner as that of \cite{m=1}, and then send it to all workers in $\left[\mathsf{N}\right]\setminus\left\{n\right\}$, yielding the same communication cost as that in \cite{m=1}, where  
\begin{itemize} 
\item when $\mathsf{K}_{\mathrm{c}}\in\left[1:\frac{\mathsf{K}}{\mathsf{N}}\right)$, 
\begin{align} 
\mathsf{R}_{\mathrm{cec}}=\mathsf{N}_{\mathrm{r}}\mathsf{K}_{\mathrm{c}}; 
\end{align}
\item when $\mathsf{K}_{\mathrm{c}}\in\left[\frac{\mathsf{K}}{\mathsf{N}}:\frac{\mathsf{K}}{\mathsf{N}}\mathsf{N}_{\mathrm{r}}\right]$, 
\begin{align}
\mathsf{R}_{\mathrm{cec}}=\frac{\mathsf{K}}{\mathsf{N}}\mathsf{N}_{\mathrm{r}};   
\end{align} 
\item when $\mathsf{K}_{\mathrm{c}}\in\left(\frac{\mathsf{K}}{\mathsf{N}}\mathsf{N}_{\mathrm{r}}:\mathsf{K}\right]$, 
\begin{align}
\mathsf{R}_{\mathrm{cec}}=\mathsf{K}_{\mathrm{c}}.    
\end{align} 
\end{itemize}
However, the benchmark scheme falls short in effectively leveraging the side information available from each worker, as the master in \cite{m=1} does not generate any message locally.

\section{Main Results} \label{Sec:Results} 
Apparently, if $\mathsf{N}_{\mathrm{r}}=1$, each worker is assigned $\frac{\mathsf{K}}{\mathsf{N}}\left(\mathsf{N}-\mathsf{N}_{\mathrm{r}}+1\right)=\mathsf{K}$ datasets and does not need to send data to other workers, hence the communication cost is $0$. In the next, we consider the case where $\mathsf{N}_{\mathrm{r}}\in\left[2:\mathsf{N}\right]$. 
%In this paper, we propose a novel achievable computing scheme to further reduce the communication cost. 
The following theorem demonstrates the performance of the proposed theorem. 
\begin{Theorem} \label{Thrm:achievability}
For the \textnormal{$\left(\mathsf{K}, \mathsf{N}, \mathsf{N}_{\mathrm{r}}, \mathsf{K}_{\mathrm{c}}, \mathsf{M}\right)$} decentralized linearly separable computation problem with \textnormal{$\mathsf{M}=\frac{\mathsf{K}}{\mathsf{N}}\left(\mathsf{N}-\mathsf{N}_{\mathrm{r}}+1\right)$}, the achieved communication cost \textnormal{$\mathsf{R}_{\mathrm{dec}}$} is given by  
%\begin{itemize} 
%\item if \textnormal{$\mathsf{N}_{\mathrm{r}}=1$}, 
%\textnormal{ 
%\begin{align}
%\mathsf{R}_{\mathrm{dec}}=0; 
%\end{align} 
%} 
%\item if \textnormal{$\mathsf{N}_{\mathrm{r}}\in\left[2:\mathsf{N}\right]$}, 
\begin{itemize} 
\item when \textnormal{$\mathsf{K}_{\mathrm{c}}\in\left[1:\frac{\mathsf{K}}{\mathsf{N}}\right)$}, 
\textnormal{
\begin{align}
\mathsf{R}_{\mathrm{dec}}=\mathsf{N}_{\mathrm{r}}\mathsf{K}_{\mathrm{c}}; 
\end{align} 
}
\item when \textnormal{$\mathsf{K}_{\mathrm{c}}\in\left[\frac{\mathsf{K}}{\mathsf{N}}:\mathsf{K}\right]$}, 
\textnormal{
\begin{align}
\mathsf{R}_{\mathrm{dec}}=\frac{\mathsf{K}}{\mathsf{N}}\mathsf{N}_{\mathrm{r}}.  
\end{align} 
}
\end{itemize} 
%\end{itemize} 
\end{Theorem} 
\begin{proof} 
When $\mathsf{K}_{\mathrm{c}}\in\left[\frac{\mathsf{K}}{\mathsf{N}}\mathsf{N}_{\mathrm{r}}\right]$, 
% we let each worker $n\in\left[\mathsf{N}\right]$ construct $X_n$ in the same manner as that of \cite{m=1}, then send it to all workers in $\left[\mathsf{N}\right]\setminus\left\{n\right\}$, hence the yielded communication cost coincides with that in \cite{m=1}. 
we opt to directly employ the benchmark scheme, resulting in the same communication cost. 
When $\mathsf{K}_{\mathrm{c}}\in\left(\frac{\mathsf{K}}{\mathsf{N}}\mathsf{N}_{\mathrm{r}}:\mathsf{K}\right]$, the detailed proof is provided in Section \ref{Sec:Scheme}. The key idea is as follows.  According to the local messages and demanded linear combinations, each worker sends the minimum number of linear combinations of messages uniformly at random to other workers.  Finally, each worker decodes the $\mathsf{K}_{\mathrm{c}}$ desired linear combinations based on the local and delivered messages. %Since the workers can use their side information to decode their demanded linear combinations, each worker is not required to transmit as many symbols to the other workers as prescribed by the benchmark scheme. In particular, it suffices to let each worker select some linear combinations of messages uniformly at random from the intersection of the linear spaces of its computed messages and the demanded linear combinations, then send them to the other workers. 
\end{proof} 
\begin{Remark} 
Our proposed scheme outperforms the benchmark scheme when \textnormal{$\mathsf{K}_{\mathrm{c}}\in\left(\frac{\mathsf{K}}{\mathsf{N}}\mathsf{N}_{\mathrm{r}}:\mathsf{K}\right]$}, 
%this is due to the benchmark scheme's inability to effectively leverage the side information available from each worker. In contrast, 
the proposed scheme fully exploits the locally computed messages of the workers such that the communication cost can be reduced, resulting a performance gain of \textnormal{$\mathsf{K}_{\mathrm{c}}-\frac{\mathsf{K}}{\mathsf{N}}\mathsf{N}_{\mathrm{r}}$}. 
\end{Remark} 
\begin{Remark} 
When \textnormal{$\mathsf{K}_{\mathrm{c}}\in\left(\frac{\mathsf{K}}{\mathsf{N}}\mathsf{N}_{\mathrm{r}}:\mathsf{K}\right]$}, there exist new challenges in the correctness proof of our proposed scheme. The main challenge is to show that under different side information across the workers and for arbitrary $\mathsf{N}-\mathsf{N}_{\mathrm{r}}$ stragglers, each worker is still capable of decoding the desired computation task from fewer received coded signals than the benchmark schemes. In Appendix \ref{Sec:Correctness}, we formally prove that with high probability the demanded linear combinations of each worker lie in the linear span of its known messages and the received coded signals.     %despite the random selections of transmitting signals,
\end{Remark} 
\begin{Theorem}[Optimality]  \label{Thrm:optimality}
For the \textnormal{$\left(\mathsf{K}, \mathsf{N}, \mathsf{N}_{\mathrm{r}}, \mathsf{K}_{\mathrm{c}}, \mathsf{M}\right)$} decentralized linearly separable computation problem with \textnormal{$\mathsf{M}=\frac{\mathsf{K}}{\mathsf{N}}\left(\mathsf{N}-\mathsf{N}_{\mathrm{r}}+1\right)$}, we have 
\textnormal{
\begin{IEEEeqnarray}{rCl}
\mathsf{R}_{\mathrm{cyc}}^*=\mathsf{R}_{\mathrm{dec}}.  
\end{IEEEeqnarray}  
} 
\end{Theorem} 
%\begin{proof} 
%%We omit the trivial case where $\mathsf{N}_{\mathrm{r}}=1$ since no transmission of data is required. When $%\mathsf{N}_{\mathrm{r}}\in\left[2:\mathsf{N}\right]$, 
%We can follow the analysis in \cite[Appendix B]{m=1} to obtain $\mathsf{R}_{\mathrm{cyc}}^*\geq\mathsf{N}_{\mathrm{r}}%\min\left\{\mathsf{K}_{\mathrm{c}}, \frac{\mathsf{K}}{\mathsf{N}}\right\}$. Note that in our problem $\mathsf{R}%_{\text{cyc}}^*\geq\mathsf{K}_{\mathrm{c}}$ is not necessary since each worker can generate some messages locally %from the datasets assigned to it. 
%\end{proof} 
\begin{proof} 
The proof follows an idea similar to that in \cite[Appendix B]{m=1}. In particular, since this paper focuses on the worst case of $\mathsf{N}-\mathsf{N}_{\mathrm{r}}$ stragglers, we would choose the set of straggler workers such that the communication cost is as large as possible. Consider worker \textnormal{$n\in\left[\mathsf{N}\right]$}  and the   \textnormal{$\frac{\mathsf{K}}{\mathsf{N}}$} messages \textnormal{$W_{n+p\mathsf{N}}, p=0, \ldots, \frac{\mathsf{K}}{\mathsf{N}}-1$}. By the cyclic assignment in \eqref{CA}, these messages  are uniquely computed by the workers in \textnormal{$\left\{n\right\}\cup\left(\left[\mathsf{N}\right]\setminus\mathcal{R}_1\right)$},  where  $\mathcal{R}_1:=\left[\mathsf{N}\right]\setminus\left\{\text{Mod}\left(n-1, \mathsf{N}\right), \ldots, \text{Mod}\left(n-\mathsf{N}+\mathsf{N}_{\mathrm{r}}, \mathsf{N}\right)\right\}$  is assumed to be the set of responding workers. 
Following the analysis in \cite[Appendix B]{m=1}, we can derive that worker $n$ needs to transmit at least \textnormal{$\min\left\{\mathsf{K}_{\mathrm{c}}, \frac{\mathsf{K}}{\mathsf{N}}\right\}\mathsf{L}$} symbols to the other workers. Summing up the transmitted symbols from all the \textnormal{$\mathsf{N}_{\mathrm{r}}$} responding workers results in \textnormal{$\mathsf{R}_{\mathrm{cyc}}^*$}. 
\end{proof}
\begin{figure}
\centering
\includegraphics[width=1.0\linewidth,scale=1]{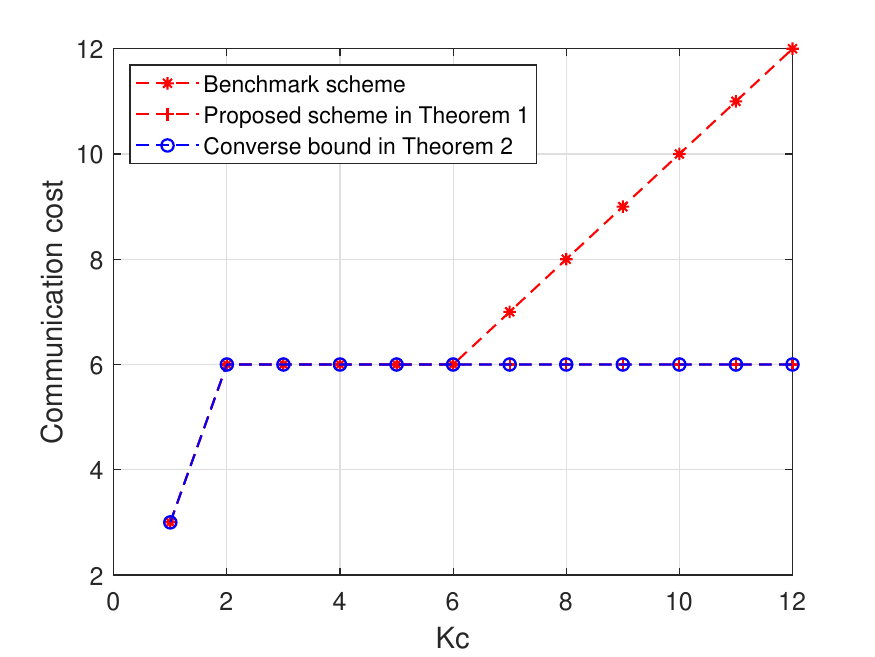}  %in itw width = 0.983
\caption{Communication costs for $\mathsf{K}=12$, $\mathsf{N}=6$, $\mathsf{N}_{\mathrm{r}}=3$.}
\label{R_Kc} %% label for entire figure
\end{figure}
\begin{figure}
\centering
\includegraphics[width=1.0\linewidth,scale=1]{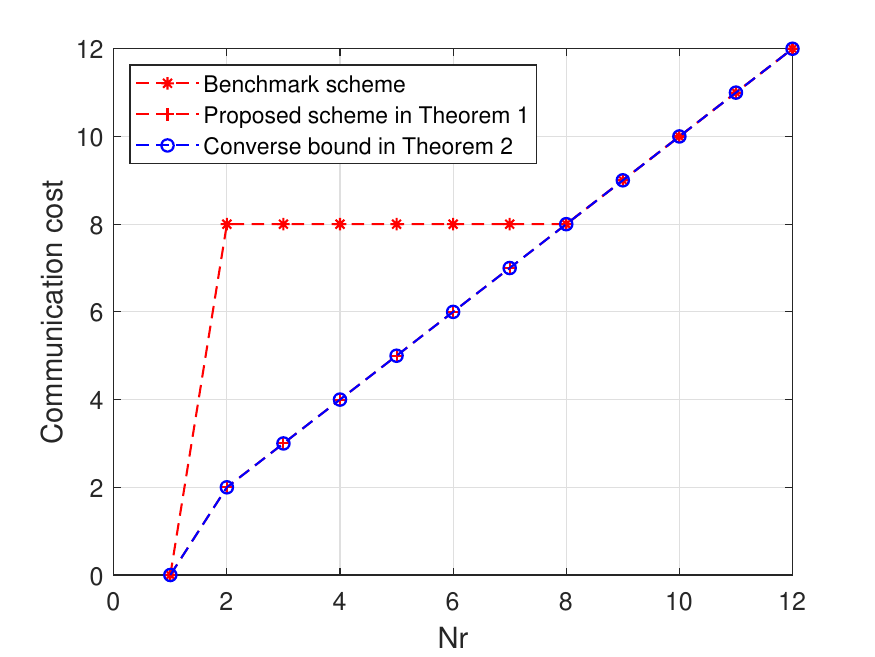}  %in itw width = 0.983
\caption{Communication costs for $\mathsf{K}=12$, $\mathsf{N}=12$, $\mathsf{K}_{\mathrm{c}}=8$.}
\label{R_Nr} %% label for entire figure
\end{figure}
\begin{Remark} 
Surprisingly, when \textnormal{$\mathsf{K}_{\mathrm{c}}\in\left[\frac{\mathsf{K}}{\mathsf{N}}\mathsf{N}_{\mathrm{r}}\right]$}, the locally computed messages of each worker cannot help reduce the communication cost compared to the centralized scheme \cite{m=1}. This means that, it suffices to adopt the benchmark scheme to achieve the optimal communication cost under the cyclic assignment. The main reason is as follows. Since the data assignment has the minimum computation cost (i.e., each dataset is assigned to $\mathsf{N}-\mathsf{N}_{\mathrm{r}}+1$ workers), to overcome the worst case of stragglers, each worker $n\in\left[\mathsf{N}\right]$ must send coded signals that carry at least $\frac{\mathsf{K}}{\mathsf{N}}$ local messages (each of which contains $\mathsf{L}$ symbols) generated from the local datasets $\mathcal{Z}_n$. In other words, the $\mathsf{N}_{\mathrm{r}}$ delivered signals carry at least $\frac{\mathsf{K}}{\mathsf{N}}\mathsf{N}_\mathrm{r}\mathsf{L}$ useful symbols. When $\mathsf{N}_{\mathrm{r}}$ is sufficiently large compared to $\mathsf{K}_{\mathrm{c}}$, the must-sent amount of delivered signals is sufficient to recover $\mathsf{K}_{\mathrm{c}}$ target linear combinations without the utilization of local messages.

  %The reason is that each worker \textnormal{$n\in\left[\mathsf{N}\right]$} has no information about the stragglers, and has to send signals in view of  the worse case, i.e., the responding workers are $\mathcal{R}_1\setminus\left\{n\right\}$ do not have any knowledge of the messages \textnormal{$W_{n+p\mathsf{N}}, p=0, \ldots, \frac{\mathsf{K}}{\mathsf{N}}-1$}, just like the master in \cite{m=1}. In other words, the workers , although they can generate some messages locally. As a result, worker $n$ must transmit the same number of symbols to the other workers as in \cite{m=1} (to the master in that context). 
% which coincides with \textnormal{$\mathsf{R}_{\mathrm{cec}}$} when \textnormal{$\mathsf{K}_{\mathrm{c}}\in\left[\frac{\mathsf{K}}{\mathsf{N}}\mathsf{N}_{\mathrm{r}}\right]$}.   
\end{Remark} 

Fig. \ref{R_Kc} depicts the relationship between communication cost and the number of linear combinations $\mathsf{K}_{\mathrm{c}}$ when $\mathsf{K}=12$, $\mathsf{N}=6$, and $\mathsf{N}_{\mathrm{r}}=3$. Clearly, when $\mathsf{K}_{\mathrm{c}}\in\left[2:6\right]$, our proposed scheme achieves the same communication cost as the benchmark scheme. When  $\mathsf{K}_{\mathrm{c}}\in\left[7:12\right]$, the proposed scheme holds an advantage over the benchmark scheme as its communication cost remains constant, whereas that of the benchmark scheme increases linearly with $\mathsf{K}_{\mathrm{c}}$. 
%for the benchmark scheme, the communication cost increases linearly with $\mathsf{K}_{\mathrm{c}}$, while for the proposed scheme, the communication cost remains unchanged as $\mathsf{K}_{\mathrm{c}}$ increases. 
Hence our proposed scheme can achieve a better performance for large $\mathsf{K}_{\mathrm{c}}$. 

Fig. \ref{R_Nr} displays the relationship between communication cost and the number of responding workers $\mathsf{N}_{\mathrm{r}}$ when $\mathsf{K}=12$, $\mathsf{N}=12$, and $\mathsf{K}_{\mathrm{c}}=8$. When $\mathsf{N}_{\mathrm{r}}=1$, the communication costs of both schemes are $0$. When $\mathsf{N}_{\mathrm{r}}\in\left[2:7\right]$, the proposed scheme achieves a smaller communication cost than the benchmark scheme; when $\mathsf{N}_{\mathrm{r}}\in\left[8:12\right]$, the two schemes have the same performance. Hence our proposed scheme is able to achieve a better performance for small $\mathsf{N}_{\mathrm{r}}$. 

Moreover, in both of the $2$ figures, the communication cost of the proposed scheme coincides with the converse bound under the cyclic assignment in Theorem \ref{Thrm:optimality}.

\section{Achievable Distributed Computing Scheme} 
\label{Sec:Scheme} 
In this section, we formally introduce the proposed computing scheme for the D2D network. Before providing the general scheme, we first present an example to illustrate the main idea. 

\begin{Example} \label{m=1exp} \textnormal{($\mathsf{K}=4, \mathsf{N}=4, \mathsf{N}_{\mathrm{r}}=3, \mathsf{K}_{\mathrm{c}}=4, \mathsf{M}=2$):} 

\;\;\textit{Data Assignment Phase:} 
\textnormal{
%In this example, each worker is assigned $\mathsf{M}=\frac{\mathsf{K}}{\mathsf{N}}\left(\mathsf{N}-\mathsf{N}_{\mathrm{r}}+\mathsf{m}\right)=2$ datasets. 
We apply the cyclic assignment, i.e.,} 
 \begin{table}[htbp]
%\caption{Data Assignment}
\begin{center}
\begin{tabular}{|c|c|c|c|}
\hline
     Worker $1$ & Worker $2$ & Worker $3$ & Worker $4$ \\  \hline 
    $D_1$ & $D_2$ & $D_3$ & $D_4$ \\  
    $D_2$ & $D_3$ & $D_4$ & $D_1$ \\ \hline 
%$r=1$ & \multicolumn{1}{c|}{906.794}    & 880.029      & \multicolumn{1}{c|}{703.829}    & 664.353      \\ \hline
%$r=3$ & \multicolumn{1}{c|}{288.669}    & 247.202      & \multicolumn{1}{c|}{264.939}    & 216.908      \\ \hline
%$r=5$ & \multicolumn{1}{c|}{167.174}    & 116.708      & \multicolumn{1}{c|}{158.307}    & 107.296      \\ \hline
\end{tabular}
\label{tab1}
\end{center}
\end{table} 

\;\;\textit{Computing Phase:} \textnormal{WLOG we let the computing task be 
\begin{IEEEeqnarray}{rCl}  
f\left(D_1, D_2, D_3, D_4\right)&&=\left[F_1; F_2; F_3; F_4\right]
=\mathbf{F}\left[W_1; W_2; W_3; W_4\right] \nonumber \\ 
&&=\begin{bmatrix} 
1 & 1 & 1 & 1 \\ 
1 & 2 & 3 & 4 \\ 
1 & 0 & 2 & 3 \\ 
1 & 2 & 1 & 4  
\end{bmatrix} 
\begin{bmatrix} 
W_1 \\ 
W_2 \\ 
W_3 \\                                 
W_4                                    
\end{bmatrix},   
\end{IEEEeqnarray}  
where
\begin{IEEEeqnarray}{rCl} 
&&F_1=W_1+W_2+W_3+W_4, \\ 
&&F_2=W_1+2W_2+3W_3+4W_4, \\ 
&&F_3=W_1+2W_3+3W_4, \\ 
&&F_4=W_1+2W_2+W_3+4W_4.   
\end{IEEEeqnarray} 
} 
\textnormal{
For simplicity, in this example, $\mathbb{F}_{\mathsf{q}}$ is assumed to be a sufficiently large prime field, which is not necessary for the general scheme where we only require that the field size $\mathsf{q}$ be sufficiently large.  
}

\textnormal{We first focus on worker $1$, who cannot compute $W_3$ and $W_4$. Let $\overline{\mathbf{F}_1}$ denote the sub-matrix of $\mathbf{F}$ comprised of the columns of $\mathbf{F}$ with indices in $\left\{3, 4\right\}$, 
%i.e., $\mathbf{F}'^{{\left\{3, 4, 7, 8, 11, 12\right\}}_{\text{c}}}$, 
which is a full-rank matrix with dimension $4\times2$. A possible vector basis for $\mathcal{N}\left(\left(\overline{\mathbf{F}_1}\right)^{T}\right)$ is $\mathbf{u}_{1, 1}=\left[0, -5, 8, -1\right]^T, \mathbf{u}_{1, 2}=\left[5, 0, -3, 1\right]^T$. Then worker $1$ computes}    
\begin{IEEEeqnarray}{rCl} 
&&0F_1-5F_2+8F_3-1F_4=2W_1-12W_2, \\ 
&&5F_1+0F_2-3F_3+1F_4=3W_1+7W_2.  
\end{IEEEeqnarray} 
%in both of which the coefficients of $W_3$ and $W_4$ are $0$. 

\textnormal{For worker $2$ who cannot compute $W_1$ and $W_4$, let $\overline{\mathbf{F}_2}$ denote the sub-matrix of $\mathbf{F}$ comprised of the columns of $\mathbf{F}$ with indices in $\left\{1, 4\right\}$.      
%i.e., $\mathbf{F}'^{{\left\{3, 4, 7, 8, 11, 12\right\}}_{\text{c}}}$, 
A possible vector basis for $\mathcal{N}\left(\left(\overline{\mathbf{F}_2}\right)^{T}\right)$ is $\mathbf{u}_{2, 1}=\left[0, -1, 0, 1\right]^T, \mathbf{u}_{2, 2}=\left[-1, -2, 3, 0\right]^T$. Then worker $2$ computes}   
\begin{IEEEeqnarray}{rCl}  
&&0F_1-1F_2+0F_3+1F_4=-2W_3, \\ 
&&-1F_1-2F_2+3F_3+0F_4=-5W_2-W_3.  
\end{IEEEeqnarray}  
%in both of which the coefficients of $W_1$ and $W_4$ are $0$. 

\textnormal{For worker $3$ who cannot compute $W_1$ and $W_2$, let $\overline{\mathbf{F}_3}$ denote the sub-matrix of $\mathbf{F}$ comprised of the columns of $\mathbf{F}$ with indices in $\left\{1, 2\right\}$.  
%i.e., $\mathbf{F}'^{{\left\{3, 4, 7, 8, 11, 12\right\}}_{\text{c}}}$, 
A possible vector basis for $\mathcal{N}\left(\left(\overline{\mathbf{F}_3}\right)^{T}\right)$ is $\mathbf{u}_{3, 1}=\left[2, 3, -1, -4\right]^T, \mathbf{u}_{3, 2}=\left[-6, -2, 3, 5\right]^T$. Then worker $3$ computes}   
\begin{IEEEeqnarray}{rCl}  
&&2F_1+3F_2-1F_3-4F_4=5W_3-5W_4, \\ 
&&-6F_1-2F_2+3F_3+5F_4=-W_3+15W_4.  
\end{IEEEeqnarray}  
%in both of which the coefficients of $W_2$ and $W_3$ are $0$. 

\textnormal{For worker $4$ who cannot compute $W_2$ and $W_3$, let $\overline{\mathbf{F}_4}$ denote the sub-matrix of $\mathbf{F}$ comprised of the columns of $\mathbf{F}$ with indices in $\left\{2, 3\right\}$.  
%i.e., $\mathbf{F}'^{{\left\{3, 4, 7, 8, 11, 12\right\}}_{\text{c}}}$, 
A possible vector basis for $\mathcal{N}\left(\left(\overline{\mathbf{F}_4}\right)^{T}\right)$ is $\mathbf{u}_{4, 1}=\left[0, -1, 1, 1\right]^T, \mathbf{u}_{4, 2}=\left[4, -4, 3, 2\right]^T$. Then worker $4$ computes }
\begin{IEEEeqnarray}{rCl}  
&&0F_1-1F_2+1F_3+1F_4=W_1+3W_4, \\ 
&&4F_1-4F_2+3F_3+2F_4=5W_1+5W_4.  
\end{IEEEeqnarray}  
%in both of which the coefficients of $W_2$ and $W_3$ are $0$. 

\textnormal{Next we show that each worker can send fewer coded messages than the benchmark scheme, while being still able to recover the target linear combinations.  Each worker $n\in\{1,2,3,4\}$ selects $\frac{\mathsf{K}}{\mathsf{N}}\left(\mathsf{N}_{\mathrm{r}}-1\right)=1$ vector $\mathbf{v}_{n, 1}$ uniformly at random from $\mathrm{span}\left\{\mathbf{u}_{n, 1}, \mathbf{u}_{n, 2}, \mathbf{u}_{n, 3}\right\}$. WLOG we let} 
\begin{IEEEeqnarray}{rCl} 
&&\mathbf{v}_{1, 1}=\mathbf{u}_{1, 1}=\left[0, -5, 8, -1\right]^T, \\            
&&\mathbf{v}_{2, 1}=\mathbf{u}_{2, 2}=\left[-1, -2, 3, 0\right]^T, \\            
&&\mathbf{v}_{3, 1}=\mathbf{u}_{3, 1}+\mathbf{u}_{3, 2}=\left[-4, 1, 2, 1\right]^T, \\            
&&\mathbf{v}_{4, 1}=2\mathbf{u}_{4, 1}+3\mathbf{u}_{4, 2}=\left[12, -14, 11, 8\right]^T.             
\end{IEEEeqnarray}        
\textnormal{Worker $n\in\left[4\right]$ then sends 
$X_n=\mathbf{v}_{n, 1}^T
\left[F_1; \cdots; F_4\right] 
$ to all workers in $\left[4\right]\setminus\left\{n\right\}$. In particular, we have} 
\begin{IEEEeqnarray}{rCl}  
\begin{bmatrix} 
X_1 \\[1pt] 
X_2 \\[1pt] 
X_3 \\[1pt] 
X_4  
\end{bmatrix} 
=\begin{bmatrix} 
\mathbf{v}_{1, 1}^T \\[3pt] 
\mathbf{v}_{2, 1}^T \\[3pt] 
\mathbf{v}_{3, 1}^T \\[3pt]  
\mathbf{v}_{4, 1}^T 
\end{bmatrix} 
\begin{bmatrix} 
F_1 \\[1pt] 
F_2 \\[1pt] 
F_3 \\[1pt] 
F_4  
\end{bmatrix} 
=\begin{bmatrix} 
2W_1-12W_2 \\[2pt] 
-5W_2-W_3 \\[2pt] 
4W_3+10W_4 \\[2pt] 
17W_1+21W_4     
\end{bmatrix}.  
\end{IEEEeqnarray}

\;\;\textit{Decoding Phase:} 
\textnormal{WLOG, we assume the set of responding workers are $\left\{1, 2, 3\right\}$. For each responding worker $i\in\{1,2,3\}$, after receiving the transmissions from other responding workers, worker $i$ has $4$ linear combinations of $F_1, \ldots, F_4$, i.e.,  
\begin{IEEEeqnarray}{rCl} \label{lcsexp1}
&&\begin{bmatrix} 
\mathbf{v}_{1, 1}^T \\ 
\vdots \\ 
\mathbf{v}_{i-1, 1}^T \\[3pt] 
\mathbf{u}_{i, 1}^T \\[3pt]   
\mathbf{u}_{i, 2}^T \\[3pt]  
\mathbf{v}_{i+1, 1}^T \\ 
\vdots \\ 
\mathbf{v}_{3, 1}^T   
\end{bmatrix} 
\begin{bmatrix} 
F_1 \\[1pt] 
F_2 \\[1pt] 
F_3 \\[1pt] 
F_4  
\end{bmatrix} 
:=\mathbf{S}_{\left[3\right]}^{\left(i\right)} 
\begin{bmatrix} 
F_1 \\[1pt] 
F_2 \\[1pt] 
F_3 \\[1pt] 
F_4  
\end{bmatrix}.  
\end{IEEEeqnarray}  
It can be checked that the matrix $\mathbf{S}_{\left[3\right]}^{\left(i\right)}$ is full-rank, for any $i\in\{1,2,3\}$. Thus worker $i\in\{1,2,3\}$ can recover the desired linear combinations by multiplying \eqref{lcsexp1} with $\left(\mathbf{S}_{\left[3\right]}^{\left(i\right)}\right)^{-1}$. 
} 

\textnormal{For worker $4$, it can decode $F_1, \ldots, F_4$ with its local content and the transmissions from any $2$ responding workers $i_1, i_2\in\{1,2,3\}$. In particular, consider the following $4$ linear combinations of $F_1, \ldots, F_4$,} 
\begin{IEEEeqnarray}{rCl} \label{lcsw4}
&&\begin{bmatrix}
\mathbf{v}_{i_1, 1}^T \\[3pt]  
\mathbf{v}_{i_2, 1}^T \\[3pt]  
\mathbf{u}_{4, 1}^T \\[3pt] 
\mathbf{u}_{4, 2}^T 
\end{bmatrix} 
\begin{bmatrix} 
F_1 \\[1pt] 
F_2 \\[1pt] 
F_3 \\[1pt] 
F_4  
\end{bmatrix} 
:=\mathbf{S}_{\left\{i_1, i_2, 4\right\}}^{\left(4\right)} 
\begin{bmatrix} 
F_1 \\[1pt] 
F_2 \\[1pt] 
F_3 \\[1pt] 
F_4   
\end{bmatrix}.  
\end{IEEEeqnarray}  
It can be checked that the matrix $\mathbf{S}_{\left\{i_1, i_2, 4\right\}}^{\left(4\right)}$ is full-rank, regardless of the choices of $i_1$ and $i_2$. Thus worker $4$ can recover the desired linear combinations by multiplying \eqref{lcsw4} with $\left(\mathbf{S}_{\left\{i_1, i_2, 4\right\}}^{\left(4\right)}\right)^{-1}$.

\;\;\textit{Performance:} 
\textnormal{The communication cost is $\frac{3\mathsf{L}}{\mathsf{L}}=3$, which coincides with the converse bound. If we directly apply the benchmark scheme, the communication cost will be $4$. Thus our scheme is able to achieve a better performance.   
}

\end{Example} 

Next, we provide the general description of our proposed scheme. We first focus on the case where $\mathsf{K}_{\mathrm{c}}\in\left(\frac{\mathsf{K}}{\mathsf{N}}\mathsf{N}_{\mathrm{r}}: \mathsf{K}\right]$. 

\;\;\emph{Data Assignment Phase:} We assign the $\mathsf{K}$ datasets to the $\mathsf{N}$ workers under the cyclic asignment in \eqref{CA}. 

\;\;\emph{Computing Phase:} 
%To meet the converse bound $\mathsf{R}_{\mathrm{cyc}}^*\geq\frac{\mathsf{K}}{\mathsf{N}}\cdot\frac{\left(\mathsf{N}_{\mathrm{r}}-\mathsf{m}\right)\mathsf{N}_{\mathrm{r}}}{\mathsf{N}_{\mathrm{r}}-1}$ when $\mathsf{K}_{\mathrm{c}}\in\left[\frac{\mathsf{K}}{\mathsf{N}}\left(\mathsf{N}_{\mathrm{r}}-\mathsf{m}\right): \mathsf{K}\right]$, we evenly and disjointly divide each message $W_k$, $k\in\left[\mathsf{K}\right]$ into $\mathsf{N}_{\mathrm{r}}-1$ sub-messages $\left\{W_{k, i}: i\in \left[\mathsf{N}_{\mathrm{r}}-1\right]\right\}$, 
%each with a length of $\frac{\mathsf{L}}{\mathsf{N}_{\mathrm{r}}-1}$, where $W_{k, i}\in\mathbb{F}_{\mathsf{q}}^{1\times\frac{\mathsf{L}}{\mathsf{N}_{\mathrm{r}}-1}}$, $i=1, \ldots, \mathsf{N}_{\mathrm{r}}-1$, and let each worker send $\frac{\mathsf{K}}{\mathsf{N}}\left(\mathsf{N}_{\mathrm{r}}-\mathsf{m}\right)$ linear combinations of sub-messages, each with a length of $\frac{\mathsf{L}}{\mathsf{N}_{\mathrm{r}}-1}$. To this end, the effective demand becomes 
We let $\overline{\mathcal{Z}_n}:=\left[\mathsf{K}\right]\setminus\mathcal{Z}_n$ denote the set of datasets not assigned to worker $n\in\left[\mathsf{N}\right]$, then let $\overline{\mathbf{F}_n}$ denote the sub-matrix of $\mathbf{F}$ comprised of the columns of $\mathbf{F}$ with indices in $\overline{\mathcal{Z}_n}$, which has a dimension of $\mathsf{K}_{\mathrm{c}}\times\frac{\mathsf{K}}{\mathsf{N}}\left(\mathsf{N}_{\mathrm{r}}-1\right)$ and is full-rank with high probability. Let $\mathbf{u}_{n, 1}, \ldots, \mathbf{u}_{n, \mathsf{K}_{\mathrm{c}}-\frac{\mathsf{K}}{\mathsf{N}}\left(\mathsf{N}_{\mathrm{r}}-1\right)}$ be a vector basis for $\mathcal{N}\left(\left(\overline{\mathbf{F}_n}\right)^{T}\right)$. Thus for each $i\in\left[\left(\mathsf{K}_{\mathrm{c}}-\frac{\mathsf{K}}{\mathsf{N}}\left(\mathsf{N}_{\mathrm{r}}-1\right)\right)\right]$, consider the following linear combination of sub-messages 
\begin{IEEEeqnarray}{rCl} \label{local} 
\mathbf{u}_{n, i}^T \mathbf{F} \left[W_1; \ldots; W_{\mathsf{K}}\right]
=\mathbf{u}_{n, i}^T \left[F_1; \ldots; F_{\mathsf{K}_{\mathrm{c}}}\right],    
\end{IEEEeqnarray}  
since $\mathbf{u}_{n, i}^T \overline{\mathbf{F}_n}=\mathbf{0}$, the linear combination in \eqref{local} is independent of any sub-message in $\left\{W_k: k\in\overline{\mathcal{Z}_n}\right\}$, thus it can be computed by worker $n$.   

 Given that the local messages of the $\mathsf{N}$ workers aid in decoding the demanded linear combinations, each worker sends the \emph{minimum} number of coded symbols (fewer than that of the benchmark scheme) to the other workers. In particular, worker $n\in\left[\mathsf{N}\right]$ selects $\frac{\mathsf{K}}{\mathsf{N}}$ vectors $\mathbf{v}_{n, 1}, \ldots, \mathbf{v}_{n, \frac{\mathsf{K}}{\mathsf{N}}}$ uniformly at random from $\mathrm{span}\left\{\mathbf{u}_{n, 1}, \ldots, \mathbf{u}_{n, \mathsf{K}_{\mathrm{c}}-\frac{\mathsf{K}}{\mathsf{N}}\left(\mathsf{N}_{\mathrm{r}}-1\right)}\right\}$, and sends $\mathbf{v}_{n, i}^T \left[F_1; \ldots; F_{\mathsf{K}_{\mathrm{c}}}\right]$, $i=1, \ldots, \frac{\mathsf{K}}{\mathsf{N}}$ to all workers in $\left[\mathsf{N}\right]\setminus\left\{n\right\}$. 

\;\;\emph{Decoding Phase:} Let $\mathcal{A}=\left\{\mathcal{A}\left(1\right), \ldots, \mathcal{A}\left(\mathsf{N}_{\mathrm{r}}\right)\right\}$ denote the set of responding workers, where $\mathcal{A}\subseteq\left[\mathsf{N}\right]$, $\lvert\mathcal{A}\rvert=\mathsf{N}_{\mathrm{r}}$.  For each responding worker $\mathcal{A}\left(i\right)$, $i\in\left[\mathsf{N}_{\mathrm{r}}\right]$, the linear combinations of $F_1, \ldots,  F_{\mathsf{K}_{\mathrm{c}}}$ it receives from workers in $\mathcal{A}\setminus\left\{\mathcal{A}\left(i\right)\right\}$, together with those generated locally, are  
\begin{IEEEeqnarray}{rCl} \label{lcs}
&&\begin{bmatrix} 
\mathbf{v}_{\mathcal{A}\left(1\right), 1}^T \\ 
\vdots \\ 
\mathbf{v}_{\mathcal{A}\left(i-1\right), \frac{\mathsf{K}}{\mathsf{N}}}^T \\[5pt] 
\mathbf{u}_{\mathcal{A}\left(i\right), 1}^T \\ 
\vdots \\ 
\mathbf{u}_{\mathcal{A}\left(i\right), \mathsf{K}_{\mathrm{c}}-\frac{\mathsf{K}}{\mathsf{N}}\left(\mathsf{N}_{\mathrm{r}}-1\right)}^T \\[1pt]  
\mathbf{v}_{\mathcal{A}\left(i+1\right), 1}^T \\ 
\vdots \\ 
\mathbf{v}_{\mathcal{A}\left(\mathsf{N}_{\mathrm{r}}\right), \frac{\mathsf{K}}{\mathsf{N}}}^T   
\end{bmatrix} 
\begin{bmatrix} 
F_1 \\ 
\vdots \\ 
F_{\mathsf{K}_{\mathrm{c}}}   
\end{bmatrix} 
:=\mathbf{S}_{\mathcal{A}}^{\left(\mathcal{A}\left(i\right)\right)} 
\begin{bmatrix} 
F_1 \\ 
\vdots \\ 
F_{\mathsf{K}_{\mathrm{c}}}   
\end{bmatrix}.  
\end{IEEEeqnarray}  

We then present the following lemma, whose proof is in Appendix \ref{Sec:Correctness}, which is the most technical part of our work. 
\begin{Lemma} \label{fullrank}
For any set $\mathcal{A}$ of responding workers, the matrix \textnormal{$\mathbf{S}_{\mathcal{A}}^{\left(\mathcal{A}\left(i\right)\right)}$} is full-rank with high probability, for any \textnormal{$i\in\left[\mathsf{N}_{\mathrm{r}}\right]$}.  
\end{Lemma} 
Thus, by Lemma \ref{fullrank}, worker $\mathcal{A}\left(i\right)$, $i\in\left[\mathsf{N}_{\mathrm{r}}\right]$ can decode the desired linear combinations by computing $\left(\mathbf{S}_{\mathcal{A}}^{\left(\mathcal{A}\left(i\right)\right)}\right)^{-1}\mathbf{S}_{\mathcal{A}}^{\left(\mathcal{A}\left(i\right)\right)}\left[F_1; \ldots; F_{\mathsf{K}_{\mathrm{c}}}\right]$. 

Moreover, for each non-responding worker $j\in\left[\mathsf{N}\right]\setminus\mathcal{A}$, it can decode $F_1, \ldots,  F_{\mathsf{K}_{\mathrm{c}}}$ with its local content and the transmissions from any $\mathsf{N}_{\mathrm{r}}-1$ responding workers $\mathcal{A}\left(i_1\right), \ldots, \mathcal{A}\left(i_{\mathsf{N}_{\mathrm{r}}-1}\right)\in\mathcal{A}$. In particular, define $\mathcal{A}':=\left\{\mathcal{A}\left(i_1\right), \ldots, \mathcal{A}\left(i_{\mathsf{N}_{\mathrm{r}}-1}\right), j\right\}$, and consider the following $\mathsf{K}_{\mathrm{c}}$ linear combinations of $F_1, \ldots,  F_{\mathsf{K}_{\mathrm{c}}}$, 
\begin{IEEEeqnarray}{rCl} \label{lcs}
&&\begin{bmatrix} 
\mathbf{v}_{\mathcal{A}\left(i_1\right), 1}^T \\ 
\vdots \\ 
\mathbf{v}_{\mathcal{A}\left(i_1\right), \frac{\mathsf{K}}{\mathsf{N}}}^T \\[5pt] 
\mathbf{v}_{\mathcal{A}\left(i_2\right), 1}^T \\ 
\vdots \\ 
\mathbf{v}_{\mathcal{A}\left(i_{\mathsf{N}_{\mathrm{r}}-1}\right), \frac{\mathsf{K}}{\mathsf{N}}}^T \\[6pt]          
\mathbf{u}_{j, 1}^T \\ 
\vdots \\ 
\mathbf{u}_{j, \mathsf{K}_{\mathrm{c}}-\frac{\mathsf{K}}{\mathsf{N}}\left(\mathsf{N}_{\mathrm{r}}-1\right)}^T \\   
\end{bmatrix} 
\begin{bmatrix} 
F_1 \\ 
\vdots \\ 
F_{\mathsf{K}_{\mathrm{c}}}   
\end{bmatrix} 
:=\mathbf{S}_{\mathcal{A}'}^{\left(j\right)} 
\begin{bmatrix} 
F_1 \\ 
\vdots \\ 
F_{\mathsf{K}_{\mathrm{c}}}  
\end{bmatrix}.  
\end{IEEEeqnarray}  
If the set of responding workers were $\mathcal{A}'$, by Lemma \ref{fullrank}, the matrix $\mathbf{S}_{\mathcal{A}'}^{\left(j\right)}$ is full-rank with high probability. Thus, each non-responding worker $j\in\left[\mathsf{N}\right]\setminus\mathcal{A}$ can decode the desired linear combinations by computing $\left(\mathbf{S}_{\mathcal{A}'}^{\left(j\right)}\right)^{-1}\mathbf{S}_{\mathcal{A}'}^{\left(j\right)}\left[F_1; \ldots; F_{\mathsf{K}_{\mathrm{c}}}\right]$. 

The decoding complexity (i.e., the number of multiplications) of each worker $n\in\left[\mathsf{N}\right]$ is $\mathcal{O}\left(\mathsf{K}_{\mathrm{c}}^2\mathsf{L}\right)$.   

\;\;\emph{Performance:} Since each worker sends $\frac{\mathsf{K}}{\mathsf{N}}$ linear combinations of the sub-messages, 
each with a length of $\mathsf{L}$, 
the required communication cost is $\frac{\mathsf{K}}{\mathsf{N}}\mathsf{N}_{\mathrm{r}}$, which coincides with the converse bound.  If we directly apply the benchmark scheme, the communication cost will be $\mathsf{K}_{\mathrm{c}}$, which is strictly larger than $\frac{\mathsf{K}}{\mathsf{N}}\mathsf{N}_{\mathrm{r}}$. Thus our  scheme can achieve a better performance.

%When $\mathsf{K}_{\mathrm{c}}\in\left(\frac{\mathsf{K}}{\mathsf{N}}\left(\mathsf{N}_{\mathrm{r}}-\mathsf{m}\right): \frac{\mathsf{K}}{\mathsf{N}}\left(\mathsf{N}_{\mathrm{r}}-\mathsf{m}+1\right)\right)$, we let each worker request $\widetilde{\mathbf{F}}\left[W_1, \ldots, W_{\mathsf{K}}\right]$, where $\widetilde{\mathbf{F}}=\left[\mathbf{F}; \mathbf{V}\right]$, in which the elements of $\mathbf{V}\in\mathbb{F}_{\mathsf{q}}^{\left(\frac{\mathsf{K}}{\mathsf{N}}\left(\mathsf{N}_{\mathrm{r}}-\mathsf{m}+1\right)-\mathsf{K}_{\mathrm{c}}\right)\times\mathsf{K}}$ are uniformly i.i.d. over $\mathbb{F}_{\mathsf{q}}$, then we can apply the scheme when $\mathsf{K}_{\mathrm{c}}=\frac{\mathsf{K}}{\mathsf{N}}\left(\mathsf{N}_{\mathrm{r}}-\mathsf{m}+1\right)$. The achieved communication cost is still $\frac{\mathsf{K}}{\mathsf{N}}\cdot\frac{\left(\mathsf{N}_{\mathrm{r}}-\mathsf{m}\right)\mathsf{N}_{\mathrm{r}}}{\mathsf{N}_{\mathrm{r}}-1}$. The decoding complexity of each worker $n\in\left[\mathsf{N}\right]$ is $\mathcal{O}\left(\mathsf{K}_{\mathrm{c}}\frac{\mathsf{K}}{\mathsf{N}}\left(\mathsf{N}_{\mathrm{r}}-\mathsf{m}+1\right)\left(\mathsf{N}_{\mathrm{r}}-1\right)^2\mathsf{L}\right)$.  

%$\overline{\mathbf{F}_{\mathcal{A}\left(i\right)}'}$ 

%%%%%
%%% Do not uncomment for double-blind review submission!
% \section*{Acknowledgments}
%
% We are indebted to Michael Shell for maintaining and improving
% \texttt{IEEEtran.cls}. 

\section{Conclusion} 
\label{Sec:Conclusion} 
In this paper, we addressed the distributed linearly separable computation problem within a fully decentralized framework, focusing on minimizing the communication cost when the computation cost is minimum. Our proposed novel distributed computing scheme effectively leverages locally computed messages from each worker, achieving optimal communication cost under the cyclic assignment. Future works encompass exploring the optimal tradeoff between computation and communication costs for this problem, and extending to the more general scenario where each worker aims to compute different linear combinations of messages.

%%%%%%
%% To balance the columns at the last page of the paper use this
%% command somewhere at the top of the first column of the last page:
%%
% \enlargethispage{-5cm} 
%%
%% where the exact amount of page reduction has to be adapted to the
%% actual situation.
%%
%% If the balancing should occur in the middle of the references, use
%% the following trigger:
%%
% \IEEEtriggeratref{3}
%%
%% which triggers a \newpage (i.e., new column) just before the given
%% reference number. Note that you need to adapt this if you modify
%% the paper. The "triggered" command can be changed if desired:
%%
% \IEEEtriggercmd{\enlargethispage{-20cm}}
%%
%%%%%%

%%%%%%
%% References:
%% We recommend the usage of BibTeX:
%%
%\bibliographystyle{IEEEtran}
%\bibliography{definitions,bibliofile}
%%
%% where we here have assume the existence of the files
%% definitions.bib and bibliofile.bib.
%% BibTeX documentation can be obtained at:
%% http://www.ctan.org/tex-archive/biblio/bibtex/contrib/doc/
%%%%%%
%% Or you use manual references (pay attention to consistency and the
%% formatting style!):

\appendices 

\section{Proof of Lemma \ref{fullrank}}
\label{Sec:Correctness}
We first present the following $2$ lemmas. 
\begin{Lemma} \label{random}
Let $\mathcal{S}$ and $\mathcal{T}$ be subspaces of \textnormal{$\mathbb{F}_{\mathsf{q}}^r$} for sufficiently large \textnormal{$\mathsf{q}$} and $\mathcal{S}\nsubseteq\mathcal{T}$. 
%Denote $d_1:=\dim\left(\mathcal{S}\right)$ and $d_2:=\dim\left(\mathcal{T}\right)$, assume that $d_2<d_1$. 
For \textnormal{$s\in\left[\dim\mathcal{S}\right]$}, Choose $s$ vectors $\mathbf{p}_1, \ldots, \mathbf{p}_s$ uniformly at random from $\mathcal{S}$, then with high probability 
\begin{enumerate} 
\item $\mathbf{p}_i\not\in\mathcal{T}$ for any $i\in\left[s\right]$; 
\item $\mathbf{p}_1, \ldots, \mathbf{p}_s$ are linearly independent.  
\end{enumerate}
%1) $\mathbf{p}_i\not\in\mathcal{T}$ for all $i\in\left[s\right]$ and 2) $\mathbf{p}_1, \ldots, \mathbf{p}_s$ are linearly independent. 
\end{Lemma} 

The proof of Lemma \ref{random} will be given in Appendix \ref{randomproof}. 

\begin{Lemma} \label{dimlowerbound}
For any \textnormal{$l\in\left[\mathsf{N}_{\mathrm{r}}\right]$} responding workers $\mathcal{A}\left(i_1\right), \ldots, \mathcal{A}\left(i_{l}\right)\in\mathcal{A}$, with high probability we have 
\textnormal{ 
\begin{align}
&\dim\mathrm{span}\left\{\mathbf{u}_{\mathcal{A}\left(i_{1}\right), 1}, \ldots, \mathbf{u}_{\mathcal{A}\left(i_{l}\right), \mathsf{K}_{\mathrm{c}}-\frac{\mathsf{K}}{\mathsf{N}}\left(\mathsf{N}_{\mathrm{r}}-1\right)}\right\} \nonumber 
\\&\geq\mathsf{K}_{\mathrm{c}}-\frac{\mathsf{K}}{\mathsf{N}}\left(\mathsf{N}_{\mathrm{r}}-l\right).  
\end{align} 
}
\end{Lemma} 

The proof of Lemma \ref{dimlowerbound} will be given in Appendix \ref{dimlowerboundproof}. 

Next, we prove the following proposition by induction. 
\begin{Proposition} \label{proposition1}
For arbitrary \textnormal{$a\in\left[\mathsf{N}_{\mathrm{r}}-1\right]$} and any \textnormal{$b\in\left(a:\mathsf{N}_{\mathrm{r}}\right]$} responding workers $\mathcal{A}\left(i_1\right), \ldots, \mathcal{A}\left(i_{b}\right)\in\mathcal{A}$, with high probability we have 
\textnormal{ 
\begin{align} \label{prop1} 
&\dim\mathrm{span}\left\{\mathbf{u}_{\mathcal{A}\left(i_1\right), 1}, \ldots, \mathbf{u}_{\mathcal{A}\left(i_a\right), \mathsf{K}_{\mathrm{c}}-\frac{\mathsf{K}}{\mathsf{N}}\left(\mathsf{N}_{\mathrm{r}}-1\right)}, 
\notag\right. 
\\&
\phantom{=\;\;}
\left.\;\;\;\;\;\;\mathbf{v}_{\mathcal{A}\left(i_{a+1}\right), 1}, \ldots, \mathbf{v}_{\mathcal{A}\left(i_b\right), \frac{\mathsf{K}}{\mathsf{N}}}\right\}\geq\mathsf{K}_{\mathrm{c}}-\frac{\mathsf{K}}{\mathsf{N}}\left(\mathsf{N}_{\mathrm{r}}-b\right).   
\end{align} 
} 
\end{Proposition}

\emph{Step 1.} We first show that when $a=\mathsf{N}_{\mathrm{r}}-1$ and $b=\mathsf{N}_{\mathrm{r}}$, Proposition \ref{proposition1} is true, i.e., with high probability 
\begin{align} \label{step1} 
&\dim\mathrm{span}\left\{\mathbf{u}_{\mathcal{A}\left(i_1\right), 1}, \ldots, \mathbf{u}_{\mathcal{A}\left(i_{\mathsf{N}_{\mathrm{r}}-1}\right), \mathsf{K}_{\mathrm{c}}-\frac{\mathsf{K}}{\mathsf{N}}\left(\mathsf{N}_{\mathrm{r}}-1\right)}, 
\notag\right. 
\\&
\phantom{=\;\;}
\left.\;\;\;\;\;\;\;\;\;\;\;\;\;\;\;\;\;\;\;\;\;\;\;\;\;\;\;\;\mathbf{v}_{\mathcal{A}\left(i_{\mathsf{N}_{\mathrm{r}}}\right), 1}, \ldots, \mathbf{v}_{\mathcal{A}\left(i_{\mathsf{N}_{\mathrm{r}}}\right), \frac{\mathsf{K}}{\mathsf{N}}}\right\}=\mathsf{K}_{\mathrm{c}}. 
\end{align} 

By Lemma \ref{dimlowerbound}, when $l=\mathsf{N}_{\mathrm{r}}-1$, we have 
\begin{align} 
&\dim\mathrm{span}\left\{\mathbf{u}_{\mathcal{A}\left(i_1\right), 1}, \ldots, \mathbf{u}_{\mathcal{A}\left(i_{\mathsf{N}_{\mathrm{r}}-1}\right), \mathsf{K}_{\mathrm{c}}-\frac{\mathsf{K}}{\mathsf{N}}\left(\mathsf{N}_{\mathrm{r}}-1\right)}\right\} 
%\\&=\dim\mathop{\sum}\limits_{j=1}^{\mathsf{N}_{\mathrm{r}}-1} \mathrm{span}\left\{\mathbf{u}_{\mathcal{A}%\left(i_j\right), 1}, \ldots, \mathbf{u}_{\mathcal{A}\left(i_j\right), \mathsf{K}_{\mathrm{c}}-\frac{\mathsf{K}}%{\mathsf{N}}\left(\mathsf{N}_{\mathrm{r}}-1\right)}\right\}
\geq\mathsf{K}_{\mathrm{c}}-\frac{\mathsf{K}}{\mathsf{N}}.  \label{nr-1}  
\end{align} 

If $\dim\mathrm{span}\left\{\mathbf{u}_{\mathcal{A}\left(i_1\right), 1}, \ldots, \mathbf{u}_{\mathcal{A}\left(i_{\mathsf{N}_{\mathrm{r}}-1}\right), \mathsf{K}_{\mathrm{c}}-\frac{\mathsf{K}}{\mathsf{N}}\left(\mathsf{N}_{\mathrm{r}}-1\right)}\right\}=\mathsf{K}_{\mathrm{c}}$, then apparently \eqref{step1} holds. 
%\begin{align}  
%&\dim\mathrm{span}\left\{\mathbf{u}_{\mathcal{A}\left(i_1\right), 1}, \ldots, \mathbf{u}_{\mathcal{A}%\left(i_{\mathsf{N}_{\mathrm{r}}-1}\right), \mathsf{K}_{\mathrm{c}}-\frac{\mathsf{K}}{\mathsf{N}}\left(\mathsf{N}%_{\text{r}}-1\right)}, 
%\notag\right. 
%\\&
%\phantom{=\;\;}
%\left.\;\;\;\;\;\;\;\;\;\;\;\;\;\;\;\;\;\;\;\;\;\;\;\;\;\;\;\;\mathbf{v}_{\mathcal{A}\left(i_{\mathsf{N}%_{\text{r}}}\right), 1}, \ldots, \mathbf{v}_{\mathcal{A}\left(i_{\mathsf{N}_{\mathrm{r}}}\right), \frac{\mathsf{K}}%{\mathsf{N}}}\right\}=\mathsf{K}_{\mathrm{c}}. 
%\end{align} 

If $\dim\mathrm{span}\left\{\mathbf{u}_{\mathcal{A}\left(i_1\right), 1}, \ldots, \mathbf{u}_{\mathcal{A}\left(i_{\mathsf{N}_{\mathrm{r}}-1}\right), \mathsf{K}_{\mathrm{c}}-\frac{\mathsf{K}}{\mathsf{N}}\left(\mathsf{N}_{\mathrm{r}}-1\right)}\right\}\in\left[\mathsf{K}_{\mathrm{c}}-\frac{\mathsf{K}}{\mathsf{N}}:\mathsf{K}_{\mathrm{c}}\right)$, it can be derived that 
\begin{align} 
&\dim\left(\mathrm{span}\left\{\mathbf{u}_{\mathcal{A}\left(i_{\mathsf{N}_{\mathrm{r}}}\right), 1}, \ldots, \mathbf{u}_{\mathcal{A}\left(i_{\mathsf{N}_{\mathrm{r}}}\right), \mathsf{K}_{\mathrm{c}}-\frac{\mathsf{K}}{\mathsf{N}}\left(\mathsf{N}_{\mathrm{r}}-1\right)}\right\}  
\notag\right. 
\\&
\phantom{=\;\;}
\left.\bigcap\;\mathrm{span}\left\{\mathbf{u}_{\mathcal{A}\left(i_1\right), 1}, \ldots, \mathbf{u}_{\mathcal{A}\left(i_{\mathsf{N}_{\mathrm{r}}-1}\right), \mathsf{K}_{\mathrm{c}}-\frac{\mathsf{K}}{\mathsf{N}}\left(\mathsf{N}_{\mathrm{r}}-1\right)}\right\}\right)  
\\&=\dim\mathrm{span}\left\{\mathbf{u}_{\mathcal{A}\left(i_{\mathsf{N}_{\mathrm{r}}}\right), 1}, \ldots, \mathbf{u}_{\mathcal{A}\left(i_{\mathsf{N}_{\mathrm{r}}}\right), \mathsf{K}_{\mathrm{c}}-\frac{\mathsf{K}}{\mathsf{N}}\left(\mathsf{N}_{\mathrm{r}}-1\right)}\right\} \nonumber 
\\&\;\;\;\;+\dim\mathrm{span}\left\{\mathbf{u}_{\mathcal{A}\left(i_1\right), 1}, \ldots, \mathbf{u}_{\mathcal{A}\left(i_{\mathsf{N}_{\mathrm{r}}-1}\right), \mathsf{K}_{\mathrm{c}}-\frac{\mathsf{K}}{\mathsf{N}}\left(\mathsf{N}_{\mathrm{r}}-1\right)}\right\} \nonumber 
\\&\;\;\;\;-\dim\mathrm{span}\left\{\mathbf{u}_{\mathcal{A}\left(i_1\right), 1}, \ldots, \mathbf{u}_{\mathcal{A}\left(i_{\mathsf{N}_{\mathrm{r}}}\right), \mathsf{K}_{\mathrm{c}}-\frac{\mathsf{K}}{\mathsf{N}}\left(\mathsf{N}_{\mathrm{r}}-1\right)}\right\}  
\\&\overset{(\text{a})}<\left(\mathsf{K}_{\mathrm{c}}-\frac{\mathsf{K}}{\mathsf{N}}\left(\mathsf{N}_{\mathrm{r}}-1\right)\right)+\mathsf{K}_{\mathrm{c}}-\mathsf{K}_{\mathrm{c}} 
\\&=\mathsf{K}_{\mathrm{c}}-\frac{\mathsf{K}}{\mathsf{N}}\left(\mathsf{N}_{\mathrm{r}}-1\right),  
\end{align}
where (a) is because when $l=\mathsf{N}_{\mathrm{r}}$, by Lemma \ref{dimlowerbound}, we have 
\begin{align} 
&\dim\mathrm{span}\left\{\mathbf{u}_{\mathcal{A}\left(i_1\right), 1}, \ldots, \mathbf{u}_{\mathcal{A}\left(i_{\mathsf{N}_{\mathrm{r}}}\right), \mathsf{K}_{\mathrm{c}}-\frac{\mathsf{K}}{\mathsf{N}}\left(\mathsf{N}_{\mathrm{r}}-1\right)}\right\} 
%\\&=\dim\mathop{\sum}\limits_{j=1}^{\mathsf{N}_{\mathrm{r}}} \mathrm{span}\left\{\mathbf{u}_{\mathcal{A}%\left(i_j\right), 1}, \ldots, \mathbf{u}_{\mathcal{A}\left(i_j\right), \mathsf{K}_{\mathrm{c}}-\frac{\mathsf{K}}%{\mathsf{N}}\left(\mathsf{N}_{\mathrm{r}}-1\right)}\right\}
=\mathsf{K}_{\mathrm{c}}.    
\end{align} 

Thus, in this case, since it can be concluded that $\mathrm{span}\left\{\mathbf{u}_{\mathcal{A}\left(i_{\mathsf{N}_{\mathrm{r}}}\right), 1}, \ldots, \mathbf{u}_{\mathcal{A}\left(i_{\mathsf{N}_{\mathrm{r}}}\right), \mathsf{K}_{\mathrm{c}}-\frac{\mathsf{K}}{\mathsf{N}}\left(\mathsf{N}_{\mathrm{r}}-1\right)}\right\}\nsubseteq\mathrm{span}\left\{\mathbf{u}_{\mathcal{A}\left(i_1\right), 1}, \ldots, \mathbf{u}_{\mathcal{A}\left(i_{\mathsf{N}_{\mathrm{r}}-1}\right), \mathsf{K}_{\mathrm{c}}-\frac{\mathsf{K}}{\mathsf{N}}\left(\mathsf{N}_{\mathrm{r}}-1\right)}\right\}$, by Lemma \ref{random}, with high probability $\mathbf{v}_{\mathcal{A}\left(i_{\mathsf{N}_{\mathrm{r}}}\right), 1}, \ldots, \mathbf{v}_{\mathcal{A}\left(i_{\mathsf{N}_{\mathrm{r}}}\right), \frac{\mathsf{K}}{\mathsf{N}}}$ are linearly independent and none of them are in $\mathrm{span}\left\{\mathbf{u}_{\mathcal{A}\left(i_1\right), 1}, \ldots, \mathbf{u}_{\mathcal{A}\left(i_{\mathsf{N}_{\mathrm{r}}-1}\right), \mathsf{K}_{\mathrm{c}}-\frac{\mathsf{K}}{\mathsf{N}}\left(\mathsf{N}_{\mathrm{r}}-1\right)}\right\}$. Therefore, by \eqref{nr-1}, \eqref{step1} holds with high probability.  
%\begin{align}  
%&\dim\mathrm{span}\left\{\mathbf{u}_{\mathcal{A}\left(i_1\right), 1}, \ldots, \mathbf{u}_{\mathcal{A}%\left(i_{\mathsf{N}_{\mathrm{r}}-1}\right), \mathsf{K}_{\mathrm{c}}-\frac{\mathsf{K}}{\mathsf{N}}\left(\mathsf{N}%_{\text{r}}-1\right)}, 
%\notag\right. 
%\\&
%\phantom{=\;\;}
%\left.\;\;\;\;\;\;\;\;\;\;\;\;\;\;\;\;\;\;\;\;\;\;\;\;\;\;\;\;\mathbf{v}_{\mathcal{A}\left(i_{\mathsf{N}%_{\text{r}}}\right), 1}, \ldots, \mathbf{v}_{\mathcal{A}\left(i_{\mathsf{N}_{\mathrm{r}}}\right), \frac{\mathsf{K}}%{\mathsf{N}}}\right\}=\mathsf{K}_{\mathrm{c}}. 
%\end{align} 

\emph{Step 2.} Suppose Proposition \ref{proposition1} is true when $a=a_0$, for any $b\in\left(a_0:\mathsf{N}_{\mathrm{r}}\right]$. 

When $a=a_0-1$, we prove the following proposition by induction. 
\begin{Proposition} \label{proposition2}
For any \textnormal{$b\in\left[a_0:\mathsf{N}_{\mathrm{r}}\right]$} responding workers $\mathcal{A}\left(i_1\right), \ldots, \mathcal{A}\left(i_{b}\right)\in\mathcal{A}$, we have 
\textnormal{ 
\begin{align} \label{prop2} 
&\dim\mathrm{span}\left\{\mathbf{u}_{\mathcal{A}\left(i_1\right), 1}, \ldots, \mathbf{u}_{\mathcal{A}\left(i_{a_0-1}\right), \mathsf{K}_{\mathrm{c}}-\frac{\mathsf{K}}{\mathsf{N}}\left(\mathsf{N}_{\mathrm{r}}-1\right)}, 
\notag\right. 
\\&
\phantom{=\;\;}
\left.\;\;\;\;\;\;\;\mathbf{v}_{\mathcal{A}\left(i_{a_0}\right), 1}, \ldots, \mathbf{v}_{\mathcal{A}\left(i_b\right), \frac{\mathsf{K}}{\mathsf{N}}}\right\}\geq\mathsf{K}_{\mathrm{c}}-\frac{\mathsf{K}}{\mathsf{N}}\left(\mathsf{N}_{\mathrm{r}}-b\right).   
\end{align} 
} 
\end{Proposition} 

\emph{Step 2a.} We first show that when $b=a_0$, Proposition \ref{proposition2} is true, i.e., with high probability 
\begin{align} \label{step1_1} 
&\dim\mathrm{span}\left\{\mathbf{u}_{\mathcal{A}\left(i_1\right), 1}, \ldots, \mathbf{u}_{\mathcal{A}\left(i_{a_0-1}\right), \mathsf{K}_{\mathrm{c}}-\frac{\mathsf{K}}{\mathsf{N}}\left(\mathsf{N}_{\mathrm{r}}-1\right)}, 
\notag\right. 
\\&
\phantom{=\;\;}
\left.\;\;\;\mathbf{v}_{\mathcal{A}\left(i_{a_0}\right), 1}, \ldots, \mathbf{v}_{\mathcal{A}\left(i_{a_0}\right), \frac{\mathsf{K}}{\mathsf{N}}}\right\}\geq\mathsf{K}_{\mathrm{c}}-\frac{\mathsf{K}}{\mathsf{N}}\left(\mathsf{N}_{\mathrm{r}}-a_0\right). 
\end{align} 

By Lemma \ref{dimlowerbound}, when $l=a_0-1$, we have 
\begin{align} 
&\dim\mathrm{span}\left\{\mathbf{u}_{\mathcal{A}\left(i_1\right), 1}, \ldots, \mathbf{u}_{\mathcal{A}\left(i_{a_0-1}\right), \mathsf{K}_{\mathrm{c}}-\frac{\mathsf{K}}{\mathsf{N}}\left(\mathsf{N}_{\mathrm{r}}-1\right)}\right\} \nonumber 
%\\&=\dim\mathop{\sum}\limits_{j=1}^{a_0-1} \mathrm{span}\left\{\mathbf{u}_{\mathcal{A}\left(i_j\right), 1}, %\ldots, \mathbf{u}_{\mathcal{A}\left(i_j\right), \mathsf{K}_{\mathrm{c}}-\frac{\mathsf{K}}{\mathsf{N}}\left(\mathsf{N}%_{\text{r}}-1\right)}\right\}
\\&\geq\mathsf{K}_{\mathrm{c}}-\frac{\mathsf{K}}{\mathsf{N}}\left(\mathsf{N}_{\mathrm{r}}-a_0+1\right).  \label{a0-1}  
\end{align} 

If $\dim\mathrm{span}\left\{\mathbf{u}_{\mathcal{A}\left(i_1\right), 1}, \ldots, \mathbf{u}_{\mathcal{A}\left(i_{a_0-1}\right), \mathsf{K}_{\mathrm{c}}-\frac{\mathsf{K}}{\mathsf{N}}\left(\mathsf{N}_{\mathrm{r}}-1\right)}\right\}\geq\mathsf{K}_{\mathrm{c}}-\frac{\mathsf{K}}{\mathsf{N}}\left(\mathsf{N}_{\mathrm{r}}-a_0\right)$, then apparently \eqref{step1_1} holds. 
%\begin{align}  
%&\dim\mathrm{span}\left\{\mathbf{u}_{\mathcal{A}\left(i_1\right), 1}, \ldots, \mathbf{u}_{\mathcal{A}%\left(i_{\mathsf{N}_{\mathrm{r}}-1}\right), \mathsf{K}_{\mathrm{c}}-\frac{\mathsf{K}}{\mathsf{N}}\left(\mathsf{N}%_{\text{r}}-1\right)}, 
%\notag\right. 
%\\&
%\phantom{=\;\;}
%\left.\;\;\;\;\;\;\;\;\;\;\;\;\;\;\;\;\;\;\;\;\;\;\;\;\;\;\;\;\mathbf{v}_{\mathcal{A}\left(i_{\mathsf{N}%_{\text{r}}}\right), 1}, \ldots, \mathbf{v}_{\mathcal{A}\left(i_{\mathsf{N}_{\mathrm{r}}}\right), \frac{\mathsf{K}}%{\mathsf{N}}}\right\}=\mathsf{K}_{\mathrm{c}}. 
%\end{align} 

If $\dim\mathrm{span}\left\{\mathbf{u}_{\mathcal{A}\left(i_1\right), 1}, \ldots, \mathbf{u}_{\mathcal{A}\left(i_{a_0-1}\right), \mathsf{K}_{\mathrm{c}}-\frac{\mathsf{K}}{\mathsf{N}}\left(\mathsf{N}_{\mathrm{r}}-1\right)}\right\}\in\left[\mathsf{K}_{\mathrm{c}}-\frac{\mathsf{K}}{\mathsf{N}}\left(\mathsf{N}_{\mathrm{r}}-a_0+1\right):\mathsf{K}_{\mathrm{c}}-\frac{\mathsf{K}}{\mathsf{N}}\left(\mathsf{N}_{\mathrm{r}}-a_0\right)\right)$, it can be derived that 
\begin{align} 
&\dim\left(\mathrm{span}\left\{\mathbf{u}_{\mathcal{A}\left(i_{a_0}\right), 1}, \ldots, \mathbf{u}_{\mathcal{A}\left(i_{a_0}\right), \mathsf{K}_{\mathrm{c}}-\frac{\mathsf{K}}{\mathsf{N}}\left(\mathsf{N}_{\mathrm{r}}-1\right)}\right\}  
\notag\right. 
\\&
\phantom{=\;\;}
\left.\bigcap\;\mathrm{span}\left\{\mathbf{u}_{\mathcal{A}\left(i_1\right), 1}, \ldots, \mathbf{u}_{\mathcal{A}\left(i_{{a_0}-1}\right), \mathsf{K}_{\mathrm{c}}-\frac{\mathsf{K}}{\mathsf{N}}\left(\mathsf{N}_{\mathrm{r}}-1\right)}\right\}\right)  
\\&=\dim\mathrm{span}\left\{\mathbf{u}_{\mathcal{A}\left(i_{a_0}\right), 1}, \ldots, \mathbf{u}_{\mathcal{A}\left(i_{a_0}\right), \mathsf{K}_{\mathrm{c}}-\frac{\mathsf{K}}{\mathsf{N}}\left(\mathsf{N}_{\mathrm{r}}-1\right)}\right\} \nonumber 
\\&\;\;\;\;+\dim\mathrm{span}\left\{\mathbf{u}_{\mathcal{A}\left(i_1\right), 1}, \ldots, \mathbf{u}_{\mathcal{A}\left(i_{a_0-1}\right), \mathsf{K}_{\mathrm{c}}-\frac{\mathsf{K}}{\mathsf{N}}\left(\mathsf{N}_{\mathrm{r}}-1\right)}\right\} \nonumber 
\\&\;\;\;\;-\dim\mathrm{span}\left\{\mathbf{u}_{\mathcal{A}\left(i_1\right), 1}, \ldots, \mathbf{u}_{\mathcal{A}\left(i_{a_0}\right), \mathsf{K}_{\mathrm{c}}-\frac{\mathsf{K}}{\mathsf{N}}\left(\mathsf{N}_{\mathrm{r}}-1\right)}\right\}  
\\&\overset{(\text{b})}<\left(\mathsf{K}_{\mathrm{c}}-\frac{\mathsf{K}}{\mathsf{N}}\left(\mathsf{N}_{\mathrm{r}}-1\right)\right)+\left(\mathsf{K}_{\mathrm{c}}-\frac{\mathsf{K}}{\mathsf{N}}\left(\mathsf{N}_{\mathrm{r}}-a_0\right)\right) \nonumber 
\\&\;\;\;\;-\left(\mathsf{K}_{\mathrm{c}}-\frac{\mathsf{K}}{\mathsf{N}}\left(\mathsf{N}_{\mathrm{r}}-a_0\right)\right)   
\\&=\mathsf{K}_{\mathrm{c}}-\frac{\mathsf{K}}{\mathsf{N}}\left(\mathsf{N}_{\mathrm{r}}-1\right),  
\end{align}
where (b) is because when $l=a_0$, by Lemma \ref{dimlowerbound}, we have 
\begin{align} 
&\dim\mathrm{span}\left\{\mathbf{u}_{\mathcal{A}\left(i_1\right), 1}, \ldots, \mathbf{u}_{\mathcal{A}\left(i_{a_0}\right), \mathsf{K}_{\mathrm{c}}-\frac{\mathsf{K}}{\mathsf{N}}\left(\mathsf{N}_{\mathrm{r}}-1\right)}\right\} \nonumber 
%\\&=\dim\mathop{\sum}\limits_{j=1}^{a_0} \mathrm{span}\left\{\mathbf{u}_{\mathcal{A}\left(i_j\right), 1}, %\ldots, \mathbf{u}_{\mathcal{A}\left(i_j\right), \mathsf{K}_{\mathrm{c}}-\frac{\mathsf{K}}{\mathsf{N}}\left(\mathsf{N}%_{\text{r}}-1\right)}\right\}
\\&\geq\mathsf{K}_{\mathrm{c}}-\frac{\mathsf{K}}{\mathsf{N}}\left(\mathsf{N}_{\mathrm{r}}-a_0\right).    
\end{align} 

Thus, in this case, since it can be concluded that $\mathrm{span}\left\{\mathbf{u}_{\mathcal{A}\left(i_{a_0}\right), 1}, \ldots, \mathbf{u}_{\mathcal{A}\left(i_{a_0}\right), \mathsf{K}_{\mathrm{c}}-\frac{\mathsf{K}}{\mathsf{N}}\left(\mathsf{N}_{\mathrm{r}}-1\right)}\right\}\nsubseteq\mathrm{span}\left\{\mathbf{u}_{\mathcal{A}\left(i_1\right), 1}, \ldots, \mathbf{u}_{\mathcal{A}\left(i_{a_0-1}\right), \mathsf{K}_{\mathrm{c}}-\frac{\mathsf{K}}{\mathsf{N}}\left(\mathsf{N}_{\mathrm{r}}-1\right)}\right\}$, by Lemma \ref{random}, with high probability $\mathbf{v}_{\mathcal{A}\left(i_{a_0}\right), 1}, \ldots, \mathbf{v}_{\mathcal{A}\left(i_{a_0}\right), \frac{\mathsf{K}}{\mathsf{N}}}$ are linearly independent and none of them are in $\mathrm{span}\left\{\mathbf{u}_{\mathcal{A}\left(i_1\right), 1}, \ldots, \mathbf{u}_{\mathcal{A}\left(i_{a_0-1}\right), \mathsf{K}_{\mathrm{c}}-\frac{\mathsf{K}}{\mathsf{N}}\left(\mathsf{N}_{\mathrm{r}}-1\right)}\right\}$. Therefore, by \eqref{a0-1}, \eqref{step1_1} holds with high probability.  

\emph{Step 2b.} Suppose Proposition \ref{proposition2} is true when $b=b_0$, i.e., with high probability 
\begin{align} \label{suppose} 
&\dim\mathrm{span}\left\{\mathbf{u}_{\mathcal{A}\left(i_1\right), 1}, \ldots, \mathbf{u}_{\mathcal{A}\left(i_{a_0-1}\right), \mathsf{K}_{\mathrm{c}}-\frac{\mathsf{K}}{\mathsf{N}}\left(\mathsf{N}_{\mathrm{r}}-1\right)}, 
\notag\right. 
\\&
\phantom{=\;\;}
\left.\;\;\;\;\mathbf{v}_{\mathcal{A}\left(i_{a_0}\right), 1}, \ldots, \mathbf{v}_{\mathcal{A}\left(i_{b_0}\right), \frac{\mathsf{K}}{\mathsf{N}}}\right\}\geq\mathsf{K}_{\mathrm{c}}-\frac{\mathsf{K}}{\mathsf{N}}\left(\mathsf{N}_{\mathrm{r}}-b_0\right).   
\end{align} 
  
When $b=b_0+1$, if 
\begin{align}
&\dim\mathrm{span}\left\{\mathbf{u}_{\mathcal{A}\left(i_1\right), 1}, \ldots, \mathbf{u}_{\mathcal{A}\left(i_{a_0-1}\right), \mathsf{K}_{\mathrm{c}}-\frac{\mathsf{K}}{\mathsf{N}}\left(\mathsf{N}_{\mathrm{r}}-1\right)}, 
\notag\right. 
\\&
\phantom{=\;\;}
\left.\!\!\!\!\mathbf{v}_{\mathcal{A}\left(i_{a_0}\right), 1}, \ldots, \mathbf{v}_{\mathcal{A}\left(i_{b_0}\right), \frac{\mathsf{K}}{\mathsf{N}}}\right\}\geq\mathsf{K}_{\mathrm{c}}-\frac{\mathsf{K}}{\mathsf{N}}\left(\mathsf{N}_{\mathrm{r}}-b_0-1\right),
\end{align}
then apparently 
\begin{align}
&\dim\mathrm{span}\left\{\mathbf{u}_{\mathcal{A}\left(i_1\right), 1}, \ldots, \mathbf{u}_{\mathcal{A}\left(i_{a_0-1}\right), \mathsf{K}_{\mathrm{c}}-\frac{\mathsf{K}}{\mathsf{N}}\left(\mathsf{N}_{\mathrm{r}}-1\right)}, 
\notag\right. 
\\&
\phantom{=\;\;}
\left.\!\!\!\!\!\!\!\!\!\mathbf{v}_{\mathcal{A}\left(i_{a_0}\right), 1}, \ldots, \mathbf{v}_{\mathcal{A}\left(i_{b_0+1}\right), \frac{\mathsf{K}}{\mathsf{N}}}\right\}\geq\mathsf{K}_{\mathrm{c}}-\frac{\mathsf{K}}{\mathsf{N}}\left(\mathsf{N}_{\mathrm{r}}-b_0-1\right). 
\end{align} 

Moreover, if 
\begin{align}
&\dim\mathrm{span}\left\{\mathbf{u}_{\mathcal{A}\left(i_1\right), 1}, \ldots, \mathbf{u}_{\mathcal{A}\left(i_{a_0-1}\right), \mathsf{K}_{\mathrm{c}}-\frac{\mathsf{K}}{\mathsf{N}}\left(\mathsf{N}_{\mathrm{r}}-1\right)}, \mathbf{v}_{\mathcal{A}\left(i_{a_0}\right), 1}, \ldots, 
\notag\right. 
\\&
\phantom{=\;\;}
\left.\!\!\!\!\!\!\!\!\!\mathbf{v}_{\mathcal{A}\left(i_{b_0}\right), \frac{\mathsf{K}}{\mathsf{N}}}\right\}\in\left[\mathsf{K}_{\mathrm{c}}-\frac{\mathsf{K}}{\mathsf{N}}\left(\mathsf{N}_{\mathrm{r}}-b_0\right):\mathsf{K}_{\mathrm{c}}-\frac{\mathsf{K}}{\mathsf{N}}\left(\mathsf{N}_{\mathrm{r}}-b_0-1\right)\right),  
\end{align} 
it can be derived that 
\begin{align} 
&\dim\left(\mathrm{span}\Big\{\mathbf{u}_{\mathcal{A}\left(i_1\right), 1}, \ldots, \mathbf{u}_{\mathcal{A}\left(i_{a_0-1}\right), \mathsf{K}_{\mathrm{c}}-\frac{\mathsf{K}}{\mathsf{N}}\left(\mathsf{N}_{\mathrm{r}}-1\right)}, 
\right. 
\\&
\left. 
\;\;\;\;\;\;\;\;\;\;\;\;\;\;\;\;\;\;\;\;\;\;\;\;\;\;\;\;\;\;\;\;\;\;\;\;\;\;\;\;\;\;\mathbf{v}_{\mathcal{A}\left(i_{a_0}\right), 1}, \ldots, \mathbf{v}_{\mathcal{A}\left(i_{b_0}\right), \frac{\mathsf{K}}{\mathsf{N}}}\Big\}  
\notag\right. 
\\&
\phantom{=\;\;}
\left. \bigcap\;\mathrm{span}\left\{\mathbf{u}_{\mathcal{A}\left(i_{b_0+1}\right), 1}, \ldots, \mathbf{u}_{\mathcal{A}\left(i_{b_0+1}\right), \mathsf{K}_{\mathrm{c}}-\frac{\mathsf{K}}{\mathsf{N}}\left(\mathsf{N}_{\mathrm{r}}-1\right)}\right\}\right)  
\\&=\dim\mathrm{span}\left\{\mathbf{u}_{\mathcal{A}\left(i_1\right), 1}, \ldots, \mathbf{u}_{\mathcal{A}\left(i_{a_0-1}\right), \mathsf{K}_{\mathrm{c}}-\frac{\mathsf{K}}{\mathsf{N}}\left(\mathsf{N}_{\mathrm{r}}-1\right)}, 
\notag\right. 
\\&
\phantom{=\;\;}
\left.\;\;\;\;\;\;\;\;\;\;\;\;\;\;\;\;\;\;\;\;\;\;\;\;\;\;\;\;\;\;\;\;\;\;\;\;\;\mathbf{v}_{\mathcal{A}\left(i_{a_0}\right), 1}, \ldots, \mathbf{v}_{\mathcal{A}\left(i_{b_0}\right), \frac{\mathsf{K}}{\mathsf{N}}}\right\} \nonumber 
\\&\;\;\;\;+\dim\mathrm{span}\left\{\mathbf{u}_{\mathcal{A}\left(i_{b_0+1}\right), 1}, \ldots, \mathbf{u}_{\mathcal{A}\left(i_{b_0+1}\right), \mathsf{K}_{\mathrm{c}}-\frac{\mathsf{K}}{\mathsf{N}}\left(\mathsf{N}_{\mathrm{r}}-1\right)}\right\} \nonumber 
\\&\;\;\;\;-\dim\mathrm{span}\left\{\mathbf{u}_{\mathcal{A}\left(i_1\right), 1}, \ldots, \mathbf{u}_{\mathcal{A}\left(i_{a_0-1}\right), \mathsf{K}_{\mathrm{c}}-\frac{\mathsf{K}}{\mathsf{N}}\left(\mathsf{N}_{\mathrm{r}}-1\right)}, 
\notag\right. 
\\&
\phantom{=\;\;}
\left.\;\;\;\;\;\;\;\;\;\;\;\;\;\;\;\mathbf{v}_{\mathcal{A}\left(i_{a_0}\right), 1}, \ldots, \mathbf{v}_{\mathcal{A}\left(i_{b_0}\right), \frac{\mathsf{K}}{\mathsf{N}}}, 
\notag\right. 
\\&
\phantom{=\;\;}
\left.\;\;\;\;\;\;\;\;\;\;\;\;\;\;\;\mathbf{u}_{\mathcal{A}\left(i_{b_0+1}\right), 1}, \ldots, \mathbf{u}_{\mathcal{A}\left(i_{b_0+1}\right), \mathsf{K}_{\mathrm{c}}-\frac{\mathsf{K}}{\mathsf{N}}\left(\mathsf{N}_{\mathrm{r}}-1\right)}\right\}  
\\&\overset{(\text{c})}<\left(\mathsf{K}_{\mathrm{c}}-\frac{\mathsf{K}}{\mathsf{N}}\left(\mathsf{N}_{\mathrm{r}}-b_0-1\right)\right)+\left(\mathsf{K}_{\mathrm{c}}-\frac{\mathsf{K}}{\mathsf{N}}\left(\mathsf{N}_{\mathrm{r}}-1\right)\right) \nonumber 
\\&\;\;\;\;-\left(\mathsf{K}_{\mathrm{c}}-\frac{\mathsf{K}}{\mathsf{N}}\left(\mathsf{N}_{\mathrm{r}}-b_0-1\right)\right)   
\\&=\mathsf{K}_{\mathrm{c}}-\frac{\mathsf{K}}{\mathsf{N}}\left(\mathsf{N}_{\mathrm{r}}-1\right),  
\end{align}
where (c) is because we have supposed that with high probability  
\begin{align} 
&\dim\mathrm{span}\left\{\mathbf{u}_{\mathcal{A}\left(i_1\right), 1}, \ldots, \mathbf{u}_{\mathcal{A}\left(i_{a_0-1}\right), \mathsf{K}_{\mathrm{c}}-\frac{\mathsf{K}}{\mathsf{N}}\left(\mathsf{N}_{\mathrm{r}}-1\right)}, 
\notag\right. 
\\&
\phantom{=\;\;}
\left.\;\;\;\;\;\;\;\;\;\;\;\;\;\;\;\;\;\mathbf{v}_{\mathcal{A}\left(i_{a_0}\right), 1}, \ldots, \mathbf{v}_{\mathcal{A}\left(i_{b_0}\right), \frac{\mathsf{K}}{\mathsf{N}}}, 
\notag\right. 
\\&
\phantom{=\;\;}
\left.\;\;\;\;\;\;\;\;\;\;\;\;\;\;\;\;\;\mathbf{u}_{\mathcal{A}\left(i_{b_0+1}\right), 1}, \ldots, \mathbf{u}_{\mathcal{A}\left(i_{b_0+1}\right), \mathsf{K}_{\mathrm{c}}-\frac{\mathsf{K}}{\mathsf{N}}\left(\mathsf{N}_{\mathrm{r}}-1\right)}\right\} \nonumber 
\\&\geq\mathsf{K}_{\mathrm{c}}-\frac{\mathsf{K}}{\mathsf{N}}\left(\mathsf{N}_{\mathrm{r}}-b_0-1\right).    
\end{align}  

Thus, in this case, since it can be concluded that with high probability $\mathrm{span}\left\{\mathbf{u}_{\mathcal{A}\left(i_{b_0+1}\right), 1}, \ldots, \mathbf{u}_{\mathcal{A}\left(i_{b_0+1}\right), \mathsf{K}_{\mathrm{c}}-\frac{\mathsf{K}}{\mathsf{N}}\left(\mathsf{N}_{\mathrm{r}}-1\right)}\right\}\nsubseteq\mathrm{span}\left\{\mathbf{u}_{\mathcal{A}\left(i_1\right), 1}, \ldots, \mathbf{u}_{\mathcal{A}\left(i_{a_0-1}\right), \mathsf{K}_{\mathrm{c}}-\frac{\mathsf{K}}{\mathsf{N}}\left(\mathsf{N}_{\mathrm{r}}-1\right)}, \mathbf{v}_{\mathcal{A}\left(i_{a_0}\right), 1}, \ldots, 
\notag\right. 
\\ 
\phantom{=\;\;}
\left.\!\!\!\!\!\!\!\!\!\!\mathbf{v}_{\mathcal{A}\left(i_{b_0}\right), \frac{\mathsf{K}}{\mathsf{N}}}\right\}$, by Lemma \ref{random}, with high probability $\mathbf{v}_{\mathcal{A}\left(i_{b_0+1}\right), 1}, \ldots, \mathbf{v}_{\mathcal{A}\left(i_{b_0+1}\right), \frac{\mathsf{K}}{\mathsf{N}}}$ are linearly independent and none of them are in $\mathrm{span}\left\{\mathbf{u}_{\mathcal{A}\left(i_1\right), 1}, \ldots, \mathbf{u}_{\mathcal{A}\left(i_{a_0-1}\right), \mathsf{K}_{\mathrm{c}}-\frac{\mathsf{K}}{\mathsf{N}}\left(\mathsf{N}_{\mathrm{r}}-1\right)}, \mathbf{v}_{\mathcal{A}\left(i_{a_0}\right), 1}, \ldots, 
\notag\right. 
\\ 
\phantom{=\;\;}
\left.\!\!\!\!\!\!\!\!\!\!\mathbf{v}_{\mathcal{A}\left(i_{b_0}\right), \frac{\mathsf{K}}{\mathsf{N}}}\right\}$. Therefore, by \eqref{suppose}, with high probability 
\begin{align}
&\dim\mathrm{span}\left\{\mathbf{u}_{\mathcal{A}\left(i_1\right), 1}, \ldots, \mathbf{u}_{\mathcal{A}\left(i_{a_0-1}\right), \mathsf{K}_{\mathrm{c}}-\frac{\mathsf{K}}{\mathsf{N}}\left(\mathsf{N}_{\mathrm{r}}-1\right)}, 
\notag\right. 
\\&
\phantom{=\;\;}
\left.\!\!\!\!\!\!\!\!\!\mathbf{v}_{\mathcal{A}\left(i_{a_0}\right), 1}, \ldots, \mathbf{v}_{\mathcal{A}\left(i_{b_0+1}\right), \frac{\mathsf{K}}{\mathsf{N}}}\right\}\geq\mathsf{K}_{\mathrm{c}}-\frac{\mathsf{K}}{\mathsf{N}}\left(\mathsf{N}_{\mathrm{r}}-b_0-1\right). 
\end{align}   

\emph{Step 2c.} Therefore, for any $b\in\left[a_0:\mathsf{N}_{\mathrm{r}}\right]$, with high probability we have 
\begin{align} 
&\dim\mathrm{span}\left\{\mathbf{u}_{\mathcal{A}\left(i_1\right), 1}, \ldots, \mathbf{u}_{\mathcal{A}\left(i_{a_0-1}\right), \mathsf{K}_{\mathrm{c}}-\frac{\mathsf{K}}{\mathsf{N}}\left(\mathsf{N}_{\mathrm{r}}-1\right)}, 
\notag\right. 
\\&
\phantom{=\;\;}
\left.\;\;\;\;\;\;\;\mathbf{v}_{\mathcal{A}\left(i_{a_0}\right), 1}, \ldots, \mathbf{v}_{\mathcal{A}\left(i_b\right), \frac{\mathsf{K}}{\mathsf{N}}}\right\}\geq\mathsf{K}_{\mathrm{c}}-\frac{\mathsf{K}}{\mathsf{N}}\left(\mathsf{N}_{\mathrm{r}}-b\right),    
\end{align} 
which proves Proposition \ref{proposition2}. 

\emph{Step 3.} Thus, for any $a\in\left[\mathsf{N}_{\mathrm{r}}-1\right]$ and any $b\in\left(a:\mathsf{N}_{\mathrm{r}}\right]$, with high probability we have 
\begin{align} 
&\dim\mathrm{span}\left\{\mathbf{u}_{\mathcal{A}\left(i_1\right), 1}, \ldots, \mathbf{u}_{\mathcal{A}\left(i_a\right), \mathsf{K}_{\mathrm{c}}-\frac{\mathsf{K}}{\mathsf{N}}\left(\mathsf{N}_{\mathrm{r}}-1\right)}, 
\notag\right. 
\\&
\phantom{=\;\;}
\left.\;\;\;\;\;\;\mathbf{v}_{\mathcal{A}\left(i_{a+1}\right), 1}, \ldots, \mathbf{v}_{\mathcal{A}\left(i_b\right), \frac{\mathsf{K}}{\mathsf{N}}}\right\}\geq\mathsf{K}_{\mathrm{c}}-\frac{\mathsf{K}}{\mathsf{N}}\left(\mathsf{N}_{\mathrm{r}}-b\right),    
\end{align} 
which proves Proposition \ref{proposition1}. 

Then by Proposition \ref{proposition1}, when $a=1$ and $b=\mathsf{N}_{\mathrm{r}}$, we have 
\begin{align} 
&\dim\mathrm{span}\left\{\mathbf{u}_{\mathcal{A}\left(i_1\right), 1}, \ldots, \mathbf{u}_{\mathcal{A}\left(i_{1}\right), \mathsf{K}_{\mathrm{c}}-\frac{\mathsf{K}}{\mathsf{N}}\left(\mathsf{N}_{\mathrm{r}}-1\right)}, 
\notag\right. 
\\&
\phantom{=\;\;}
\left.\;\;\;\;\;\;\;\;\;\;\;\;\;\;\;\;\;\;\;\;\;\;\;\;\;\;\;\;\;\mathbf{v}_{\mathcal{A}\left(i_{2}\right), 1}, \ldots, \mathbf{v}_{\mathcal{A}\left(i_{\mathsf{N}_{\mathrm{r}}}\right), \frac{\mathsf{K}}{\mathsf{N}}}\right\}=\mathsf{K}_{\mathrm{c}},  
\end{align} 
which completes the proof of Lemma \ref{fullrank}. 

\section{Proof of Lemma \ref{random}} \label{randomproof} 
Let $d_1 := \dim\mathcal{S}$ and $d_2 := \dim \left(\mathcal{S} \cap \mathcal{T}\right)$. Since $\mathcal{S} \not\subseteq \mathcal{T}$, we have $\mathcal{S} \cap \mathcal{T} \subsetneq \mathcal{S}$, hence $d_2 < d_1$. We choose the vectors $\mathbf{p}_1, \dots, \mathbf{p}_s$ one at a time, and estimate the probability that the required conditions are satisfied inductively. For simplicity, we always denote $ \mathcal{V}_i = \mathrm{span} \left\{ \mathbf{p}_1, \dots, \mathbf{p}_i \right\}$ for each $i\in\left[s\right]$. 

\emph{Step 1.} It is clear that 
\begin{align} 
&		 \lvert \mathcal{S} \rvert = \mathsf{q}^{d_1} \quad \text{and} \quad \lvert \mathcal{S} \cap \mathcal{T} \rvert = \mathsf{q}^{d_2} \leq \mathsf{q}^{d_1-1}. 
\end{align} 
		Therefore the probability for a randomly chosen $\mathbf{p}_1 \in \mathcal{S}$ not to be in $\mathcal{T}$ is 
\begin{align} 
		 \mathrm{Pr}\left(\mathbf{p}_1 \notin \mathcal{T} \mid \mathbf{p}_1 \in \mathcal{S}\right) &= 1- \frac{\mathsf{q}^{d_2}}{\mathsf{q}^{d_1}} 
\\&\geq 1- \frac{\mathsf{q}^{d_1-1}}{\mathsf{q}^{d_1}} 
\\&= 1-\frac{1}{\mathsf{q}} 
\\&> 1-\frac{2}{\mathsf{q}}. 
\end{align} 

\emph{Step 2.} For the choice of each subsequent $\mathbf{p}_i \in \mathcal{S}$, to guarantee the linear independence of $\left\{ \mathbf{p}_1, \dots, \mathbf{p}_i \right\}$, we need to make sure that $\mathbf{p}_i$ is not in the linear span of the previously chosen $i-1$ vectors. Therefore we estimate the probably that $\mathbf{p}_i \notin \left( \mathcal{T} \cup \mathcal{V}_{i-1} \right)$. Since we have 
\begin{align} 
		 \lvert \mathcal{S} \cap \mathcal{T} \rvert = \mathsf{q}^{d_2} \leq \mathsf{q}^{d_1-1} \quad \text{and} \quad \lvert \mathcal{V}_{i-1} \rvert \leq \mathsf{q}^{i-1} \leq \mathsf{q}^{d_1-1}, 
\end{align} 
		it follows that 
\begin{align} 
		 \lvert \left( \mathcal{S} \cap \mathcal{T} \right) \cup \mathcal{V}_{i-1} \rvert \leq \lvert \mathcal{S} \cap \mathcal{T} \rvert + \lvert \mathcal{V}_{i-1} \rvert \leq 2 \mathsf{q}^{d_1-1}, 
\end{align} 
		therefore 
\begin{align} 
		 \mathrm{Pr}( \mathbf{p}_i \notin (\mathcal{T} \cup \mathcal{V}_{i-1}) \mid \mathbf{p}_i \in \mathcal{S} ) \geq 1 - \frac{2\mathsf{q}^{d_1-1}}{\mathsf{q}^{d_1}} = 1 - \frac{2}{\mathsf{q}}. 
\end{align} 
		
		\emph{Step 3.} To summarize, we choose the vectors $\mathbf{p}_1, \dots, \mathbf{p}_s \in \mathcal{S}$ in the listed order. For the choice of each $\mathbf{p}_i$, the probability for the required condition to hold is at least $1-\frac{2}{\mathsf{q}}$. Therefore the probability for all vectors to meet the required conditions is at least $\left( 1 - \frac{2}{\mathsf{q}} \right)^s$, which converges to $1$ as $\mathsf{q} \to \infty$. 

Hence, we prove Lemma \ref{random}.

\section{Proof of Lemma \ref{dimlowerbound}} \label{dimlowerboundproof}  
Note that proving $\dim\mathrm{span}\left\{\mathbf{u}_{\mathcal{A}\left(i_{1}\right), 1}, \ldots, \mathbf{u}_{\mathcal{A}\left(i_{l}\right), \mathsf{K}_{\mathrm{c}}-\frac{\mathsf{K}}{\mathsf{N}}\left(\mathsf{N}_{\mathrm{r}}-1\right)}\right\}\\ \geq\mathsf{K}_{\mathrm{c}}-\frac{\mathsf{K}}{\mathsf{N}}\left(\mathsf{N}_{\mathrm{r}}-l\right)$   
is equivalent to proving $\dim\left(\mathcal{C}\left(\overline{\mathbf{F}_{\mathcal{A}\left(i_1\right)}}\right)\cap\cdots\cap\mathcal{C}\left(\overline{\mathbf{F}_{\mathcal{A}\left(i_l\right)}}\right)\right)\leq\frac{\mathsf{K}}{\mathsf{N}}\left(\mathsf{N}_{\mathrm{r}}-l\right)$, which will be shown by contradiction. 

Assume that with high probability $\dim\left(\mathcal{C}\left(\overline{\mathbf{F}_{\mathcal{A}\left(i_1\right)}}\right)\cap\cdots\cap\mathcal{C}\left(\overline{\mathbf{F}_{\mathcal{A}\left(i_l\right)}}\right)\right)=h>\frac{\mathsf{K}}{\mathsf{N}}\left(\mathsf{N}_{\mathrm{r}}-l\right)$. Let $\overline{\mathbf{f}_{\mathcal{A}\left(i_j\right), 1}}, \ldots, \overline{\mathbf{f}_{\mathcal{A}\left(i_j\right), \frac{\mathsf{K}}{\mathsf{N}}\left(\mathsf{N}_{\mathrm{r}}-1\right)}}$ be the columns of $\overline{\mathbf{F}_{\mathcal{A}\left(i_j\right)}}$, for each $j\in\left[l\right]$. Thus $\mathcal{C}\left(\overline{\mathbf{F}_{\mathcal{A}\left(i_1\right)}}\right)\cap\cdots\cap\mathcal{C}\left(\overline{\mathbf{F}_{\mathcal{A}\left(i_l\right)}}\right)$ is spanned by $h$ linearly independent combinations of $\overline{\mathbf{f}_{\mathcal{A}\left(i_j\right), 1}}, \ldots, \overline{\mathbf{f}_{\mathcal{A}\left(i_j\right), \frac{\mathsf{K}}{\mathsf{N}}\left(\mathsf{N}_{\mathrm{r}}-1\right)}}$, for any $j\in\left[l\right]$. Specifically, we let  
\begin{align}
&\left[\overline{\mathbf{f}_{\mathcal{A}\left(i_1\right), 1}}, \ldots, \overline{\mathbf{f}_{\mathcal{A}\left(i_1\right), \frac{\mathsf{K}}{\mathsf{N}}\left(\mathsf{N}_{\mathrm{r}}-1\right)}}\right] 
\left[\boldsymbol{\alpha}_{\mathcal{A}\left(i_1\right), 1}, \ldots, \boldsymbol{\alpha}_{\mathcal{A}\left(i_1\right), h}\right]=\cdots \nonumber 
\\&=\left[\overline{\mathbf{f}_{\mathcal{A}\left(i_l\right), 1}}, \ldots, \overline{\mathbf{f}_{\mathcal{A}\left(i_l\right), \frac{\mathsf{K}}{\mathsf{N}}\left(\mathsf{N}_{\mathrm{r}}-1\right)}}\right] 
\left[\boldsymbol{\alpha}_{\mathcal{A}\left(i_l\right), 1}, \ldots, \boldsymbol{\alpha}_{\mathcal{A}\left(i_l\right), h}\right],  
\end{align} 
where $\boldsymbol{\alpha}_{\mathcal{A}\left(i_j\right), t}\in\mathbb{F}_{\mathsf{q}}^{\frac{\mathsf{K}}{\mathsf{N}}\left(\mathsf{N}_{\mathrm{r}}-1\right)\times1}$ for any $j\in\left[l\right]$ and $t\in\left[h\right]$, and $\boldsymbol{\alpha}_{\mathcal{A}\left(i_j\right), 1}, \ldots, \boldsymbol{\alpha}_{\mathcal{A}\left(i_j\right), h}$ are linearly independent, for any $j\in\left[l\right]$.  

Let $\mathcal{B}$ be any subset of $\left[\mathsf{K}_{\mathrm{c}}\right]$ such that $\lvert\mathcal{B}\rvert=\frac{\mathsf{K}}{\mathsf{N}}\mathsf{N}_{\mathrm{r}}$, denote $\overline{\mathbf{F}_{\mathcal{A}\left(i_j\right)}^{\left(1\right)}}$ as the sub-matrix of $\overline{\mathbf{F}_{\mathcal{A}\left(i_j\right)}}$ comprised of the rows of $\overline{\mathbf{F}_{\mathcal{A}\left(i_j\right)}}$ with indices in $\mathcal{B}$, $j=1, \ldots, l$, and let $\overline{\mathbf{f}_{\mathcal{A}\left(i_j\right), 1}^{\left(1\right)}}, \ldots, \overline{\mathbf{f}_{\mathcal{A}\left(i_j\right), \frac{\mathsf{K}}{\mathsf{N}}\left(\mathsf{N}_{\mathrm{r}}-1\right)}^{\left(1\right)}}$ be the columns of $\overline{\mathbf{F}_{\mathcal{A}\left(i_j\right)}^{\left(1\right)}}$, for each $j\in\left[l\right]$. Therefore, the following equality holds:   
\begin{align}
&\left[\overline{\mathbf{f}_{\mathcal{A}\left(i_1\right), 1}^{\left(1\right)}}, \ldots, \overline{\mathbf{f}_{\mathcal{A}\left(i_1\right), \frac{\mathsf{K}}{\mathsf{N}}\left(\mathsf{N}_{\mathrm{r}}-1\right)}^{\left(1\right)}}\right] 
\left[\boldsymbol{\alpha}_{\mathcal{A}\left(i_1\right), 1}, \ldots, \boldsymbol{\alpha}_{\mathcal{A}\left(i_1\right), h}\right]=\cdots \nonumber 
\\&=\left[\overline{\mathbf{f}_{\mathcal{A}\left(i_l\right), 1}^{\left(1\right)}}, \ldots, \overline{\mathbf{f}_{\mathcal{A}\left(i_l\right), \frac{\mathsf{K}}{\mathsf{N}}\left(\mathsf{N}_{\mathrm{r}}-1\right)}^{\left(1\right)}}\right] 
\left[\boldsymbol{\alpha}_{\mathcal{A}\left(i_l\right), 1}, \ldots, \boldsymbol{\alpha}_{\mathcal{A}\left(i_l\right), h}\right].   
\end{align} 
In other words,  $\mathcal{C}\left(\overline{\mathbf{F}_{\mathcal{A}\left(i_1\right)}^{\left(1\right)}}\right)\cap\cdots\cap\mathcal{C}\left(\overline{\mathbf{F}_{\mathcal{A}\left(i_l\right)}^{\left(1\right)}}\right)$ is spanned by at least $h$ linearly independent combinations of $\overline{\mathbf{f}_{\mathcal{A}\left(i_j\right), 1}^{\left(1\right)}}, \ldots, \overline{\mathbf{f}_{\mathcal{A}\left(i_j\right), \frac{\mathsf{K}}{\mathsf{N}}\left(\mathsf{N}_{\mathrm{r}}-1\right)}^{\left(1\right)}}$, for any $j\in\left[l\right]$, thus we have $\dim\left(\mathcal{C}\left(\overline{\mathbf{F}_{\mathcal{A}\left(i_1\right)}^{\left(1\right)}}\right)\cap\cdots\cap\mathcal{C}\left(\overline{\mathbf{F}_{\mathcal{A}\left(i_l\right)}^{\left(1\right)}}\right)\right)\geq h$. Let $\mathbf{u}_{\mathcal{A}\left(i_j\right), 1}^{\left(1\right)}, \ldots, \mathbf{u}_{\mathcal{A}\left(i_j\right), \frac{\mathsf{K}}{\mathsf{N}}}^{\left(1\right)}$ be a vector basis for $\mathcal{N}\left(\left(\overline{\mathbf{F}_{\mathcal{A}\left(i_j\right)}^{\left(1\right)}}\right)^{T}\right)$, for each $j\in\left[l\right]$, hence it can be derived that  
\begin{align} 
&\dim\mathrm{span}\left\{\mathbf{u}_{\mathcal{A}\left(i_1\right), 1}^{\left(1\right)}, \ldots, \mathbf{u}_{\mathcal{A}\left(i_l\right), \frac{\mathsf{K}}{\mathsf{N}}}^{\left(1\right)}\right\} 
\\&=\frac{\mathsf{K}}{\mathsf{N}}\mathsf{N}_{\mathrm{r}}-\dim\left(\mathcal{C}\left(\overline{\mathbf{F}_{\mathcal{A}\left(i_1\right)}^{\left(1\right)}}\right)\cap\cdots\cap\mathcal{C}\left(\overline{\mathbf{F}_{\mathcal{A}\left(i_l\right)}^{\left(1\right)}}\right)\right) 
\\&=\frac{\mathsf{K}}{\mathsf{N}}\mathsf{N}_{\mathrm{r}}-h 
\\&<\frac{\mathsf{K}}{\mathsf{N}}l. \label{contradiction} 
\end{align}

However, it can be inferred from \cite[Lemma 2]{m=1} that with high probability $\dim\mathrm{span}\left\{\mathbf{u}_{\mathcal{A}\left(i_1\right), 1}^{\left(1\right)}, \ldots, \mathbf{u}_{\mathcal{A}\left(i_l\right), \frac{\mathsf{K}}{\mathsf{N}}}^{\left(1\right)}\right\}$ is exactly $\frac{\mathsf{K}}{\mathsf{N}}l$, which contradicts with \eqref{contradiction}. 

Therefore, with high probability $\dim\left(\mathcal{C}\left(\overline{\mathbf{F}_{\mathcal{A}\left(i_1\right)}}\right)\cap\cdots\cap\mathcal{C}\left(\overline{\mathbf{F}_{\mathcal{A}\left(i_l\right)}}\right)\right)\leq\frac{\mathsf{K}}{\mathsf{N}}\left(\mathsf{N}_{\mathrm{r}}-l\right)$, and consequently, 
\begin{align}
&\dim\mathrm{span}\left\{\mathbf{u}_{\mathcal{A}\left(i_{1}\right), 1}, \ldots, \mathbf{u}_{\mathcal{A}\left(i_{l}\right), \mathsf{K}_{\mathrm{c}}-\frac{\mathsf{K}}{\mathsf{N}}\left(\mathsf{N}_{\mathrm{r}}-1\right)}\right\} \nonumber 
\\&\geq\mathsf{K}_{\mathrm{c}}-\frac{\mathsf{K}}{\mathsf{N}}\left(\mathsf{N}_{\mathrm{r}}-l\right).  
\end{align} 

Hence, we prove Lemma \ref{dimlowerbound}.

\end{document}